\documentclass[authorversion,sigconf]{acmart}

\AtBeginDocument{%
  \providecommand\BibTeX{{%
    \normalfont B\kern-0.5em{\scshape i\kern-0.25em b}\kern-0.8em\TeX}}}

\copyrightyear{2022} 
\acmYear{2022} 
\setcopyright{acmlicensed}\acmConference[FAccT '22]{2022 ACM Conference on Fairness, Accountability, and Transparency}{June 21--24, 2022}{Seoul, Republic of Korea}
\acmBooktitle{2022 ACM Conference on Fairness, Accountability, and Transparency (FAccT '22), June 21--24, 2022, Seoul, Republic of Korea}
\acmPrice{15.00}
\acmDOI{10.1145/3531146.3534645}
\acmISBN{978-1-4503-9352-2/22/06}

\setcopyright{acmcopyright}
\copyrightyear{2022}
\acmYear{2022}

\usepackage[utf8]{inputenc}
\usepackage[english]{babel}
\usepackage{multirow}
\usepackage{mathtools}
\usepackage{subcaption}
\usepackage{amsthm}
\newtheorem{theorem}{Theorem}

\newtheorem{lemma}[theorem]{Lemma}
\theoremstyle{definition}
\newtheorem{definition}[theorem]{Definition}
\DeclareMathOperator*{\argmax}{arg\,max}

\begin{document}

\title[Enforcing Group Fairness in Algorithmic Decision Making]{Enforcing Group Fairness in Algorithmic Decision Making: Utility Maximization Under Sufficiency}

\author{Joachim Baumann}
\email{baumann@ifi.uzh.ch}
\orcid{0000-0003-2019-4829}
\affiliation{%
  \institution{University of Zurich}
  \city{Zurich}
  \country{Switzerland}
}
\affiliation{%
  \institution{Zurich University of Applied Sciences}
  \city{Zurich}
  \country{Switzerland}
}

\author{Anik\'{o} Hann\'{a}k}
\email{hannak@ifi.uzh.ch}
\orcid{0000-0002-0612-6320}
\affiliation{%
  \institution{University of Zurich}
  \city{Zurich}
  \country{Switzerland}
}

\author{Christoph Heitz}
\email{christoph.heitz@zhaw.ch}
\orcid{0000-0002-6683-4150}
\affiliation{%
  \institution{Zurich University of Applied Sciences}
  \city{Zurich}
  \country{Switzerland}
}

\renewcommand{\shortauthors}{Baumann et al.}

\begin{abstract}
Binary decision making classifiers are not fair by default. Fairness requirements are an additional element to the decision making rationale, which is typically driven by maximizing some utility function. In that sense, algorithmic fairness can be formulated as a constrained optimization problem. This paper contributes to the discussion on how to implement fairness, focusing on the fairness concepts of positive predictive value (PPV) parity, false omission rate (FOR) parity, and sufficiency (which combines the former two).

We show that group-specific threshold rules are optimal for PPV parity and FOR parity, similar to well-known results for other group fairness criteria. However, depending on the underlying population distributions and the utility function, we find that sometimes an upper-bound threshold rule for one group is optimal: utility maximization under PPV parity (or FOR parity) might thus lead to selecting the individuals with the smallest utility for one group, instead of selecting the most promising individuals. This result is counter-intuitive and in contrast to the analogous solutions for statistical parity and equality of opportunity.

We also provide a solution for the optimal decision rules satisfying the fairness constraint sufficiency. We show that more complex decision rules are required and that this leads to within-group unfairness for all but one of the groups. We illustrate our findings based on simulated and real data.

\end{abstract}



\maketitle

\section{Introduction}

Advances in machine learning (ML) have led to a rise in algorithmic decision making systems that assist or replace humans to make consequential decisions.
Today, such algorithms are used in various domains, such as credit lending~\cite{Fuster2017,kozodoi2022credit-scoring}, pretrial detention~\cite{angwin2016machine}, hiring~\cite{Millera}, and many more.
It has been shown that this often violates fairness across protected groups~\cite{barocas-hardt-narayanan}.
This is especially worrying if the prediction-based decision systems systematically harm marginalized groups, and, in particular, if they are applied in domains where a decision is potentially life-changing for the affected individuals~\cite{10.2307/24758720}.
A potential way to reduce ML-based discrimination is to mitigate outcome disparities across some predefined groups~\cite{hardt2016equality,10.1145/3097983.3098095,Corbett-Davies2018,pmlr-v81-menon18a,pmlr-v81-dwork18a,lipton2018does,barocas-hardt-narayanan}.
In order to measure and eventually correct for these disparities, different mathematical notions of so-called \textit{group fairness metrics} have been proposed~\cite{narayanan2018translation,verma2018fairness}.
The group fairness metrics that have attracted the most interest are independence, separation, and sufficiency~\cite{barocas-hardt-narayanan}.
These three definitions of fairness are all ``entirely reasonable and desirable''~\cite{Kearns2019EthicalAlgorithm}, however, they are mutually exclusive except for in highly constrained cases, which are unlikely to occur in practice~\cite{Kleinberg2016,Friedler2016,Chouldechova2017,Pleiss2017}.
Hence, decision makers must choose one metric over the others.

In this paper, we focus on the fairness of prediction-based decision systems that take decisions based on the prediction of a variable $Y$, which is unknown at the time of decision making.
Different methods have been developed to ensure the fairness of such systems, most of which fall into one of three categories: pre-processing, in-processing, or post-processing~\cite{pessach2020algorithmic,caton2020fairness}.
One line of papers within the field of algorithmic fairness is concerned with optimal decision rules satisfying some group fairness constraint~\cite{hardt2016equality,10.1145/3097983.3098095,lipton2018does,10.1145/3097983.3098095,pmlr-v81-menon18a}.
Thereby, the prediction model is treated as given, but the decision maker has the freedom to modify the decision rule for fulfilling fairness constraints, i.e., the predictions are post-processed so that the resulting decisions are fair w.r.t. a specified protected attribute.
Following this approach, we formulate the goal of fairness as a constrained optimization problem where a standard approach is to assume that the decision maker's goal is to maximize some utility%
\footnote{
We will define the term utility and formulate the constrained (as well as the unconstrained) optimization problem in Section~\ref{ssec:OptimalityofDecisionRulesWithorWithoutFairnessConstraints}.
}
function while also satisfying some fairness constraint~\cite{mitchell2021algorithmic}.
Such optimal decision rules have been derived for the group fairness metrics statistical parity, conditional statistical parity, TPR parity, and FPR parity~\cite{hardt2016equality,10.1145/3097983.3098095,lipton2018does}.
It has been shown that optimal decision rules that satisfy these fairness constraints are characterized by lower-bound threshold rules.%
\footnote{
In the fair ML literature, so-called \textit{thresholding} is arguably the most typical decision rule for probabilistic classifiers, also because of its conceptual similarity to the way humans take decisions~\cite{kleinberg2017humandecisions,caton2020fairness}.
In this paper, we refer to this type of decision rule as a \textit{lower-bound threshold rule}.
}
Surprisingly, to our knowledge, no such solution has been derived for the group fairness metrics PPV parity, FOR parity, and sufficiency.
This paper closes this gap by deriving optimal decision rules for these group fairness metrics.
Our main contributions and findings are:
\begin{itemize}
    \item We show that optimal decision rules satisfying PPV parity or FOR parity take the form of group-specific (lower-bound or upper-bound) thresholds.
    \item We find that, surprisingly, under PPV parity or FOR parity, it can be optimal for decision makers to apply an upper-bound threshold for one group (depending on the populations and the applied utility function). In such situations, the most promising individuals are left out, leading to an extreme form of within-group unfairness.
    \item We provide a solution for the optimal decision rules that satisfy sufficiency as the combination of both PPV parity and FOR parity. We find that this definition of fairness requires more complex decision rules (i.e., decision rules that do not take the form of a simple lower- or upper-bound threshold) and leads to within-group unfairness for all but one of the groups.
    \item We highlight the trade-off between fairness across groups and within groups. 
\end{itemize}

The remainder of the paper is structured as follows:
Section~\ref{sec:related-work} introduces the most important group fairness metrics and provides the necessary background.
In Section~\ref{sec:optimal-decisions-under-fairness-constraints}, we formalize the (un)constrained optimization problem and solve it for several group fairness metrics.
Section~\ref{sec:IllustrativeExamples} demonstrates the solutions for optimal decision rules under these fairness constraints based on simulated and real data.
Section~\ref{sec:Conclusions} concludes the paper.

\section{Related Work}
\label{sec:related-work}

\subsection{Group Fairness Metrics}
\label{ssec:Related-Work-Group-Fairness-Metrics}

Much of the technical literature on algorithmic fairness strives to create some generalized notion of fairness in terms of the impact an algorithm has on different groups~\cite{Dwork2012,10.1145/2783258.2783311,10.2307/24758720,Chouldechova2017,10.1145/3038912.3052660,Corbett-Davies2018,berk2021criminal}.
As ML algorithms are used more and more for consequential decision making, their impact on individuals and groups may be tremendous.
Numerous metrics have been suggested to quantify the group fairness of decision making algorithms~\cite{narayanan2018translation}.
Most of these group fairness criteria fall into one of three categories: independence, separation, or sufficiency~\cite{barocas-hardt-narayanan}.
Table~\ref{tab:Groupfairnessmetrics} provides the mathematical definitions for those three criteria%
\footnote{
See Section~\ref{ssec:ProblemStatementandNotation} for a description of the notations used for the equations.
}.

Independence -- also called \textit{statistical parity}~\cite{Dwork2012} -- compares decision rates across groups (i.e., the fraction of individuals who are granted a loan in each group), whereas the other two criteria compare error rates across groups~\cite{verma2018fairness}.
Conditional statistical parity extends this definition of fairness by allowing a set of legitimate features to affect the decision~\cite{Kamiran2013,10.1145/3097983.3098095}.
True positive rate (TPR) parity -- also called \textit{equal opportunity}~\cite{hardt2016equality} -- and false positive rate (FPR) parity are relaxations of the separation criterion.
Positive predictive value (PPV) parity -- also called \textit{predictive parity}~\cite{Chouldechova2017} -- and false omission rate (FOR) parity are relaxations of the sufficiency criterion -- which has also been called \textit{conditional use accuracy equality} by~\cite{berk2021criminal} or \textit{overall predictive parity} by~\cite{mayson2018bias}.
There is an essential difference between separation and sufficiency:
TPR and FPR focus on a subpopulation that is defined by $Y$.
In contrast, PPV (also called \textit{precision}) and FOR focus on a subpopulation that is defined by $D$.%
\footnote{
All four metrics can be expressed by their respective complements:
PPV parity is equivalent to false discovery rate parity, FOR parity is equivalent to negative predictive value parity, TPR parity is equivalent to false negative rate parity, and FPR parity corresponds to true negative rate parity.
}
In the loan granting scenario, the TPR denotes the fraction of those individuals who are granted a loan from all those who would not default.
For the PPV, on the other hand, only those individuals who are granted a loan are considered to measure the fraction of individuals who repay it.

PPV parity, FOR parity, and sufficiency are relevant notions of fairness, not only theoretically but also in practice.
Most prominent is probably the case of the 2016 debate surrounding the tool COMPAS (which gives judges recidivism risk predictions that are supposed to inform them on whether or not a defendant should be released in different stages of the criminal justice system), where~\cite{angwin2016machine} published an article saying that the tool systematically disadvantages black defendants because of a FPR disparity.
However, Northpointe (the developers of COMPAS) responded that the two metrics TPR parity and FPR parity are not appropriate for assessing recidivism risk scales and that instead PPV parity and FOR parity are appropriate criteria~\cite{Dieterich2016}.
They conclude that their tool is not unfair because it satisfies those two metrics.
In addition to recidivism prediction, PPV parity is also prevalent in predictive policing~\cite{Simoiu2017} (where the metric is usually called \textit{hit rate}) and in personalized online ads (where the notion of \textit{click through rates}~\cite{wang2011click}, which is an equivalent metric, is omnipresent).
\begin{table*}
  \caption{Group fairness metrics. The acronyms stand for true positive rate (TPR), false positive rate (FPR), positive predictive value (PPV), and false omission rate (FOR).}
  \label{tab:Groupfairnessmetrics}
  \begin{tabular}{@{}lll@{}}
    \toprule
        \textbf{Fairness criterion}&\textbf{Parity metric}&\textbf{Equation}\\
    \midrule
        Independence                 & Statistical parity   & $P[D =1|A = 0] = P[D = 1|A = 1]$                \\ \hline
        \multirow{2}{*}{Separation}  & TPR parity           & $P[D = 1|Y = 1, A = 0] = P[D = 1|Y = 1, A = 1]$ \\ 
                                     & FPR parity           & $P[D = 1|Y = 0, A = 0] = P[D = 1|Y = 0, A = 1]$ \\ \hline
        \multirow{2}{*}{Sufficiency} & PPV parity           & $P[Y = 1|D = 1, A = 0] = P[Y = 1|D = 1, A = 1]$ \\ 
                                     & FOR parity           & $P[Y = 1|D = 0, A = 0] = P[Y = 1|D = 0, A = 1]$ \\ 
    \bottomrule
\end{tabular}
\end{table*}

Another often discussed statistical concept in algorithmic fairness studies is calibration, which is defined as $P[Y=1|S=s] = s$, where $s$ denotes a real-valued score~\cite{Kleinberg2016,Pleiss2017,Chouldechova2017}.
An extended notion of calibration that also accounts for group membership is provided by~\cite{barocas-hardt-narayanan}.
They call it calibration by group and formally define it as $P[Y=1|S=s,A=0]=P[Y=1|S=s,A=1] = s$.
This notion of fairness is closely related to sufficiency, which is why some confusion regarding the differences between calibration and sufficiency (or one of its relaxations) emerged.
\cite{liu19implicit-criterion} and~\cite{barocas-hardt-narayanan} state that unconstrained learning satisfies group calibration and the fairness metric sufficiency.
In contrast,~\cite{Chouldechova2017} claims that it is possible that calibration is satisfied while PPV parity is not.
\cite{garg2020fairness} clarify this confusion by pointing out the difference between these two metrics: As calibration is defined for every score $s$ (which is assumed to be a continuous value and not a binary one), whereas PPV parity is just measured for a binary outcome, the two notions of fairness cannot be used interchangeably.
In particular, they show that for groups with different probability distributions, calibration does not necessarily imply sufficiency.
In this work, we investigate group fairness metrics regarding a protected attribute that divides individuals into groups with different probability distributions.

\subsection{Optimal Decisions and Fairness}
\label{ssec:Optimal-Decisions-and-Fairness}

Much of the extensive literature on algorithmic fairness is concerned with mitigating ML-based discrimination across protected groups.
According to~\cite{mitchell2021algorithmic}, a standard way of ensuring algorithmic fairness is to formulate it as a constrained optimization problem.
Thereby, a specific kind of utility function is maximized while also satisfying a fairness constraint~\cite{hardt2016equality,10.1145/3097983.3098095,lipton2018does,pmlr-v81-menon18a}.
This approach allows a utility-maximizing decision maker to derive optimal fair decision rules.
Absent any fairness constraint, applying a uniform threshold to all groups is optimal~\cite{Corbett-Davies2018}.
However, this does not automatically lead to fair decisions w.r.t. specific groups~\cite{barocas-hardt-narayanan}.
Due to the mathematical incompatibility of most group fairness metrics~\cite{Kleinberg2016,Friedler2016,Chouldechova2017,Pleiss2017}, the constrained optimization problem must be solved separately for any chosen definition of fairness.
This has been done for some group fairness metrics but not for others:
\cite{hardt2016equality} and~\cite{10.1145/3097983.3098095} have shown that optimal decision rules that satisfy (conditional) statistical parity, TPR parity, and FPR parity take the form of group-specific lower-bound thresholds.
Several other scholars have investigated thresholding solutions, such as~\cite{fish2016,valera2018enhancing,pmlr-v81-menon18a,lipton2018does}.
However, to our knowledge, a solution for the optimization problem satisfying PPV parity, FOR parity, or sufficiency does not yet exist.
This paper closes this research gap by providing a solution for deriving optimal decision rules that satisfy one of these three group fairness metrics.

In the computer science and in philosophical literature, sufficiency (or one of its relaxations, PPV parity and FOR parity) is often mentioned as one of the main fairness metrics~\cite{Chouldechova2017,narayanan2018translation,verma2018fairness,Kearns2019EthicalAlgorithm,barocas-hardt-narayanan,pessach2020algorithmic,caton2020fairness,Leben2020,berk2021criminal,Makhlouf2021Applicability,Kuppler2021,lee2021fair-selective-classification,baumann2022SDS_fairness_principle}.
Several algorithmic fairness papers have studied sufficiency or one of its relaxations.
\cite{Kasy2021} use an economic approach to argue that PPV parity is insufficient for fairness as it does not question existing differences between or within groups.
\cite{canetti2019soft} explore the possibilities of satisfying several fairness constraints at once, namely, parity of PPV, FOR, TPR, and FPR, but they do not provide a solution for PPV parity or FOR parity alone.
However, none of these authors derive optimal decision rules that satisfy (one of) these fairness constraints.
Such a solution is crucial to know what decision rational decision makers take if any of these group fairness metrics are enforced.

\section{optimal decisions under fairness constraints}
\label{sec:optimal-decisions-under-fairness-constraints}

This section provides a theoretical solution to maximizing the decision maker's utility while satisfying a group fairness definition (PPV parity, FOR parity, or sufficiency).
In the following, we first state the problem and introduce general notations before introducing an additional notion of fairness called \textit{within-group fairness}, which will prove to be helpful for the interpretation of the theoretical results.
Then, we formulate the optimization problem to be solved (with and without fairness constraints) in Subsection~\ref{ssec:OptimalityofDecisionRulesWithorWithoutFairnessConstraints}, before actually solving it for three specific group fairness definitions (see Subsections~\ref{ssec:ppv_parity} and~\ref{sec:sufficiency} and Appendix~\ref{appendix:for_parity}).

\subsection{Problem Statement and Notations}
\label{ssec:ProblemStatementandNotation}

Let us first introduce the specific context of our work, along with the main assumptions and some notations.
We assume a decision maker has to make a binary decision $D$ for each individual $i$, based on a feature vector $x_i \in \mathbb{R}^m$, which includes a protected attribute $a_i \in A$, denoting the group membership (sometimes also called \textit{sensitive attribute}).
Let $n_{A=a}$ be the number of individuals that are part of a group $a$.
Following related work, we restrict our analysis to a binary protected attribute $A$.
However, our analysis generalizes to all cases with a discrete protected attribute with more than two values.
An example may be the decision of a bank to grant a loan, based on $x_i$.%
\footnote{
In the loan granting scenario, $x_i$ might include an applicant's bill-paying history, unpaid debt, or past foreclosures.
}
We assume that the decisive feature for the decision is a binary target variable $Y$.
For a perfect predictor, every individual that belongs to the positive class ($Y=1$) must receive decision $D=1$, and vice versa~\cite{murphy2012machine,mitchell2021algorithmic}.
However, $Y$ is unknown at the time of decision making and is replaced by the probability $p_i=P[Y=1]$, which is given as a function of $x_i$, provided by a probabilistic prediction algorithm.
Generalizing the idea of a perfect predictor to probabilities means that individuals with a higher $p_i$ should be assigned $D=1$ and individuals with a lower probability of belonging to the positive class should receive the decision $D=0$.
The decision rule is thus a function $d$ that maps $p_i$ (and, possibly, $a_i$) to a binary decision.%
\footnote{
Notice that changing this decision rule represents a form of post-processing~\cite{mehrabi2019survey}.
There is no need to know the specific features used to train an algorithm because the learned model is treated as a black box.
}
Similar to~\cite{10.1145/3097983.3098095,canetti2019soft}, for our analysis, we assume furthermore that each group's probability distribution has strictly positive density.

In this paper, we formulate algorithmic fairness as a constrained optimization problem.
The goal of a rational decision maker is to maximize the expected utility while also satisfying some definition of group fairness.
In this section, we solve this constrained optimization problem for the group fairness definitions PPV parity, FOR parity, and sufficiency.

\subsection{Within-Group Fairness}

Before we derive the solution for a utility-maximizing decision maker that must satisfy some group fairness metric, let us formally define another notion of fairness -- which will be helpful for the interpretation of the theoretical results.
\begin{definition}
(Within-group fairness).
We say that a decision rule $d(p,a)$ satisfies \textit{within-group fairness} with respect to protected attribute $A$ if $\forall i \in S_{a|D=1} \forall j \in S_{a|D=0} (p_i > p_j)$, where $S_a$ is the set of all individuals of group $a \in A$.
\end{definition}
\noindent
Decision rules satisfying within-group fairness ensure that, within each group, a larger probability $p$ always leads to a higher chance of $D=1$.
More specifically, no individual that is assigned $D=1$ has a lower probability than any of the individuals that are given $D=0$.
In contrast, within-group unfairness results if there is at least one pair of individuals $(i,j)$ where $p_i$ is larger then $p_j$, and still $D_i=0$ and $D_j=1$. 
As more or less such cases can exist, there are different degrees to which within-group fairness can be violated.
We say that a decision rule $d(p,a)$ leads to an extreme form of within-group unfairness if $\forall i \in S_{a|D=1} \forall j \in S_{a|D=0} (p_i < p_j)$.
This is equivalent to applying an upper-bound threshold.

Within-group fairness requires that, within a group, individuals with a higher probability of belonging to the positive class ($p[Y=1]$) should have a higher chance of being assigned $D=1$ than individuals with a lower probability of belonging to the positive class.
For the loan example, it would be viewed as unfair if a loan is granted to one person but denied to another person with a higher probability of paying back the loan.
In many applications, such a perspective can be morally justified. 
Similarly, in the context of COMPAS, it is morally just to detain a defendant ($D=1$) with a very high risk of committing a violent crime if released ($Y=1$).
As we will see in more detail below, optimal decision rules satisfying PPV parity, FOR parity, or sufficiency do not always satisfy this notion of fairness.

\subsection{Optimal Decision Rules With or Without Fairness Constraints}
\label{ssec:OptimalityofDecisionRulesWithorWithoutFairnessConstraints}

For our theoretical analysis, we assume that a rational decision maker relies on a prediction model to cope with the uncertainty of the decision-relevant variable $Y$.
We assume that if $Y$ was known, the decision would be given.
More specifically, we assume that the decision maker's choice would be $D=1$ in the case of $Y=1$ and vice versa~\cite{murphy2012machine}.
However, in most real-world scenarios, a perfect predictor does not exist, which introduces uncertainty regarding the outcome of a decision.
There are four possible outcomes, all of which can be weighted according to the decision maker's desirability, representing a standard approach in the fair ML literature~\cite{mitchell2021algorithmic}.
This leads to the following expected individual utility%
\footnote{
For example, $u_{21}$ denotes the utility of making a decision $D=0$ and having outcome $Y=1$ occur, a so-called false negative (see Appendix~\ref{appendix:Utility-weighted-confusion-matrix}).
The definition of these utilities is context-specific.
In many cases, it would be straightforward for the decision maker to estimate them.
For example, a bank can easily calculate its utility in terms of monetary gains or losses for a successful loan (as opposed to a default) based on interest rates.
}:
\begin{equation}
u_i=\begin{cases}
u_{11} p_i + u_{12} (1-p_i), & \text{for $D=1$}\\
u_{21} p_i + u_{22} (1-p_i), & \text{for $D=0$}.
\end{cases}
\label{equation:individualutility}
\end{equation}
Defining $\tilde u_i$ as the expected {\em relative utility gain} when switching the decision from $D=0$ to $D=1$ gives $\tilde u_i =0$ for $D=0$, and  $\tilde u_i = \alpha p_i + \beta (1-p_i)$ for $D=1$, with the two parameters $\alpha= u_{11} - u_{21}$ and $\beta= u_{12} - u_{22}$.
It can be shown easily that maximizing $u_i$ is equivalent to maximizing $\tilde u_i$.
Moreover, the above made assumption that $Y=1$ implies $D=1$ requires that $\alpha>\beta$.

We assume that the decision maker takes not only one decision $d$, but many decisions $d_i$, over a population of individuals (e.g., when making loan decisions for many applicants).
In this case, the goal of a rational decision maker is to maximize the total expected utility $\tilde U$, which leads to the following optimization problem:
\begin{equation}
\argmax_d \;\;\;\; \tilde U=\sum_{i \in S} \tilde u_i d_i = \sum_{i \in S} \left( p_i (\alpha - \beta) + \beta \right) d_i ,
\label{equation:optimizationproblem}
\end{equation}%
where $S$ is the set of all individuals and $d_i$ is a binary multiplier representing the decision that is made for an individual $i$.
The optimum unconstrained decision rule $d^{\ast}$ is thus:
\begin{equation}
d^{\ast}_i=\begin{cases}
    1, & \text{for $p_i > \frac{- \beta}{\alpha - \beta} $} \\
    0, & \text{otherwise}
\end{cases}
\label{equation:optimal_unconstrained_d_rule}
\end{equation}
and takes the form of a single lower-bound threshold.
In the following, we interpret the decision problem as a \emph{selection problem}, denoting individuals with $D=1$ as ``being selected.''

The unconstrained solution does not ensure fairness w.r.t. the protected attribute at all and, in fact, is likely to produce unfairness (as measured with different group fairness metrics, see Section~\ref{ssec:Related-Work-Group-Fairness-Metrics}).
Decision makers who want to maximize their utility while taking fair decisions must solve the following constrained optimization problem:
\begin{equation}
\argmax_d \;\;\;\; \tilde U \;\;\; \text{ subject to } \; \textit{some fairness constraint}.
\label{equation:optimizationproblem-constrained}
\end{equation}%
As we outlined in Section~\ref{ssec:Optimal-Decisions-and-Fairness}, this constrained optimization problem has been solved for some group fairness metrics%
\footnote{
The authors of~\cite{hardt2016equality} use a function they call \textit{immediate utility} and~\cite{10.1145/3097983.3098095} rely on \textit{loss minimization}.
Both approaches can easily be formulated in terms of what we call decision maker utility, which is why the solutions of~\cite{hardt2016equality} and~\cite{10.1145/3097983.3098095} also hold in this setting.
}
(statistical parity, conditional statistical parity, TPR parity, and FPR parity) but not for others, such as PPV parity, FOR parity, or sufficiency.
In the remainder of this chapter, we solve the constrained optimization problem stated in Equation~\ref{equation:optimizationproblem-constrained} (using the three mentioned group fairness metrics as fairness constraints) for two different cases:
\begin{align*}
&\text{case I) } \text{the number of individuals to be selected (}n_{D=1} \text{) is}\\ 
&\quad \text{predefined,}\\
&\text{case II) } \text{the number of individuals to be selected (}n_{D=1} \text{) is} \\ 
&\quad \text{not predefined.}\\
\end{align*}

\subsection{Optimal Decision Rules under PPV Parity}
\label{ssec:ppv_parity}

We now present the optimal solution for the optimization problem stated in Equation~\ref{equation:optimizationproblem-constrained} constrained by the group fairness metric positive predictive value (PPV) parity for both cases I and II.
The PPV is defined as the average probability of individuals with $D=1$ to have $Y=1$, which can be written as $\frac{1}{n_{D=1}} \sum\limits_{i \in S} p_i d_i$.
The fairness definition PPV parity requires this value to be the same across groups.
Thus, the constrained optimization problem has the form:
\begin{equation}
\begin{split}
&\argmax_d \;\;\;\; \tilde U=\sum_{i \in S} \left( p_i (\alpha - \beta) + \beta \right) d_i \\
&\text{subject to} \;\;\, \frac{1}{n_{A=0|D=1}} \sum_{j \in S_0} p_j d_j = \frac{1}{n_{A=1|D=1}} \sum_{j \in S_1} p_j d_j = PPV, \\
& \;\;\;\;\;\;\;\;\;\;\;\;\;\;\;\;\; \text{for } PPV \in [0,1] ,
\end{split}
\label{equation:optimizationproblem_equivalent_with_PPV}
\end{equation}%
where $S_a$ is the set of all individuals of group $a$ and $n_{A=a|D=1}$ denotes the number of individuals in group $a$ with $D=1$.
Each decision rule results in a specific selection of individuals, which also yields a specific selection for each group $S_a$.
Since the PPV can only be defined if at least one individual is selected, we assume $n_{A=a|D=1}\geq1$ for each group.

We derive the solution to this optimization problem in two consecutive steps.
\begin{itemize}
    \item First, we derive the optimal decision rules $d^{\ast}$ for a simplified constraint: We assume that the PPV of both groups must be equal to a predefined value $PPV_t \in [0, 1]$.
    \item Second, we solve the full optimization problem by maximizing the decision maker's utility over all possible values of $PPV_t$.
\end{itemize}

We now derive the solution for the first step, thus specifying a value $PPV_t \in [0,1]$ for the constraint.
We do this under the assumption of a positive probability density of individuals over the full range $[0,1]$ for both groups, and in the limit case of very large populations ($n_{A=a} \rightarrow \infty$).
Thus, for each $PPV_t$, there exist individuals in each group with $p=PPV_t$.%
\footnote{
This technical assumption simplifies the notation.
For finite group sizes, the equality constraint in Equation~\ref{equation:optimizationproblem_equivalent_with_PPV} may not be met precisely for many values of $PPV_t$, and the fairness constraint might only be fulfilled approximately. Thus, the equality requirement of the FC has to be softened into approximate equality. However, the proofs are also valid for an approximate version of equality.
}
The most straightforward selection fulfilling the fairness constraint thus consists of selecting one of these individuals in each group.
Obviously, other selections exist, for example selecting more than one individual with $p=PPV_t$, or selecting individuals in an interval $[PPV_t-\epsilon,PPV_t+\epsilon]$ such that the average $p$ of the selection equals $PPV_t$.
However, many other selection rules are conceivable, with different numbers of selected individuals. 

For a predefined number of selected individuals $n_{D=1}$ (i.e., case I), the following Lemma holds:
\begin{lemma}
\label{lemma:case_I_PPV}
For a given value of $PPV_t$ and a predefined number of selected individuals $n_{D=1}$, any selection fulfilling the fairness constraint of Equation~\ref{equation:optimizationproblem_equivalent_with_PPV} leads to a total utility $\tilde U$ of:
\begin{equation}
\tilde U = (\alpha PPV_t + \beta (1-PPV_t) ) n_{D=1}.
\label{equation:rearrangedpredefined}
\end{equation}
\end{lemma}
Lemma~\ref{lemma:case_I_PPV} (proof in Appendix~\ref{appendix:proof-lemma-case_I_PPV}) shows that the fairness constraint already defines the total utility, if $n_{D=1}$ is given.
In other words: any decision rule $d(p,a)$ with $n_{D=1}$ that satisfies the constraint stated in Equation~\ref{equation:optimizationproblem_equivalent_with_PPV} for a given $PPV_t$ is optimal.
We thus end up with two independent selection problems, one for each group, which consists of finding a selection of individuals characterized by the fact that their average probability equals $PPV_t$.
For each group $a$, selections with different numbers $n_{A=a|D=1}$ are possible.
As long as the predefined $n_{D=1}$ is met, the group membership of the selected individuals does not matter for the resulting total utility.
Hence, there may be several solutions to the optimization problem that differ regarding the number of individuals selected per group (i.e., representing different combinations of ($n_{A=0|D=1},n_{A=1|D=1}$)), with $n_{A=0|D=1}+n_{A=1|D=1} = n_{D=1}$.
Note that most of these solutions violate the group fairness metric statistical parity while still meeting the fairness criterion of PPV parity.

We now analyze case II, where $n_{D=1}$ is not predefined.
Lemma~\ref{lemma:case_I_PPV} also leads to another important result: For values $PPV_t$ for which $\alpha PPV_t + \beta (1-PPV_t)<0$, a decision maker who wants to maximize the total utility should minimize $n_{D=1}$, thus selecting only one individual from each group, yielding a total utility of $\tilde U = 2 (\alpha PPV_t + \beta (1-PPV_t) )$ for a binary protected attribute.
In the following, we thus assume that $\alpha PPV_t + \beta (1-PPV_t)>0$.
Again we assume that the size of both groups is large but finite.
Lemma~\ref{lemma:case_I_PPV} shows that, under these assumptions, the decision maker's goal is to find the selection that satisfies the constraint $PPV=PPV_t$ with the maximum $n_{D=1}$.
Theorem~\ref{theorem:mainresult_optimalsolutionwiththresholdrule} specifies the solution of this optimization problem (which can be solved independently for each group):
\begin{theorem}
\label{theorem:mainresult_optimalsolutionwiththresholdrule}
For any given $PPV_t$, the optimal fair decision rules $d^{\ast}$ (i.e., decision rules that maximize $\tilde U$ while satisfying $PPV=PPV_t$) take the following form:
\protect\begin{equation}
d^{\ast}_i=\protect\begin{cases}
    \protect\begin{rcases}
    1, & \text{for $p_i \geq \tau_a$} \\
    0, & \text{otherwise}
    \protect\end{rcases}
    \text{for $PPV_t > BR_{A=a}$}\\
    \protect\begin{rcases}
    1, & \text{for $p_i \leq \tau_a$} \\
    0, & \text{otherwise}
    \protect\end{rcases}
    \text{for $PPV_t< BR_{A=a}$},
\protect\end{cases}
\label{equation:d_rule_notpredefined_PPV}
\protect\end{equation}
where $\tau_a$ denote different group-specific constants and $BR_{A=a}$ denotes group $a$'s base rate (BR) which is defined as the ratio of individuals belonging to the positive class ($Y=1$) in a group: $BR_{A=a} = P[Y=1|A=a] = \frac{1}{n_{A=a}} \sum\limits_{i \in S_a} p_i$.
\end{theorem}
\begin{proof}
We begin with the case $PPV_t<BR_{A=a}$. We define a group-specific function $g_1(n_{A=a|D=1})$, defined as the minimum value of PPV among all decision rules $\vec d$ with a specified $n_{A=a|D=1}$, i.e., $g_1(n_{A=a|D=1})= \min\limits_{\vec d} \frac{1}{n_{A=a|D=1}} \sum p_i d_i$. Obviously, $g_1(n_{A=a|D=1})$ is given by selecting the $n_{A=a|D=1}$ individuals with the smallest values of $p$. The function $g_1(n)$ for $n=1,...,n_{A=a}$ is monotonously increasing, with $g_1(1)=0$
\footnote{Recall that we consider a limit of very large populations, so the individual with the lowest $p_i$ is characterized by $p_i=0$. For $n=1$, the minimum PPV value is achieved by selecting just this individual.}
and $g_1(n_{A=a})=BR_{A=a}$. 
It is now easy to see that solving the equation $g_1(n)=PPV_t$ w.r.t. $n$ yields the maximum possible value $n$ that meets the PPV condition: Assume that there was a value $m>n$ for which a decision rule exists such that $PPV=PPV_t$. As $g_1$ is monotonically increasing, this implies $m \leq n$, which is a contradiction. Thus, for the case $PPV_t<BR_{A=a}$, the maximum achievable $n_{A=a|D=1}$ with $\frac{1}{n_{A=a|D=1}} \sum p_i d_i = PPV_t$ in the space of all possible decision rules is achieved by selecting all individuals with $p_i\leq\tau_a$. The corresponding upper-bound threshold $\tau_a$ is given by the unique solution of $g_1(n)=PPV_t$.

For $PPV_t>BR_{A=a}$, an analogous argumentation holds by introducing a function $g_2(n_{A=a|D=1})= \max\limits_{\vec d} \frac{1}{n_{A=a|D=1}} \sum p_i d_i$. This is a monotonically decreasing function with $g_2(1)=1$ and $g_2(n_{A=a})=BR_{A=a}$. The unique solution of $g_2(n)=PPV_t$ yields the lower-bound threshold $\tau_a$ that meets the PPV condition.
\end{proof}

Finally, we perform the second step of the solution: from a discretization of all $PPV$, for which a solution exists, we choose the one that (in combination with the corresponding $n_{D=1}$) maximizes the total utility.
Thereby, every $n_{D=1}$ is composed of the optimal selections $n_{A=a|D=1}$ for all groups $a \in A$, as elaborated in the first step of the solution.

We provide an analogous solution for the optimal decision rules satisfying FOR parity in Appendix~\ref{appendix:for_parity}.

\subsection{Optimal Decision Rules under Sufficiency}
\label{sec:sufficiency}

Based on the solutions presented above, we now describe the decision rules that maximize the decision maker's utility while satisfying sufficiency (requiring PPV parity and FOR parity).
This gives the constrained optimization problem:
\begin{equation}
\begin{split}
&\argmax_d \;\;\;\; \tilde U = \sum_{i \in S} \tilde u_i \\
&\text{subject to} \;\;\;\; \;\;\;\; \;\;\;   \frac{1}{n_{A=0|D=1}} \sum_{j \in S_0} p_j d_j = \frac{1}{n_{A=1|D=1}} \sum_{j \in S_1} p_j d_j \\
& \;\;\;\; \;\;\;\; \;\;\;\; \;\;\;\; \,  \frac{1}{n_{A=0|D=0}} \sum_{j \in S_0} p_j (1-d_j) = \frac{1}{n_{A=1|D=0}} \sum_{j \in S_1} p_j (1-d_j),
\end{split}
\label{equation:optimizationproblem_sufficiency}
\end{equation}%
where the first constraint represents PPV parity and the second constraint ensures FOR parity.
Similar to our PPV parity solution, we also proceed in two steps for optimal decision rules satisfying sufficiency.
First, we derive the optimal decision rules for a given value of $PPV=PPV_t$.
Second, we solve the optimization problem by choosing a PPV-FOR combination that maximizes the decision maker's utility.

We start with an optimal decision rule satisfying PPV parity (see Equation~\ref{equation:d_rule_notpredefined_PPV}) and then add the second constraint (requiring FOR parity).
Recall that a decision rule splits this group into those selected ($D=1$) and those not selected ($D=0$).
Thus, we can write:
\begin{equation}
\sum_{i \in S_a} p_i = \left( \sum_{i \in S_a} p_i (1-d_i) \right) + \left( \sum_{i \in S_a} p_i d_i \right).
\label{equation:BR-is-PPF-and-FOR}
\end{equation}
As we specified $PPV_{A=a}=PPV_t$, PPV parity is satisfied.
Thus, this gives:
\begin{equation}
n_{A=a} BR_{A=a} = n_{A=a|D=0} FOR_{A=a} + n_{A=a|D=1} PPV_t.
\label{equation:PPV-parity-BR-equation}
\end{equation}
With $n_{A=a|D=0} = n_{A=a} - n_{A=a|D=1}$ and some reformulation, we get:
\begin{equation}
FOR_{A=a} = \frac{n_{A=a} BR_{A=a} - n_{A=a|D=1} PPV_t}{n_{A=a} - n_{A=a|D=1}}.
\label{equation:sufficiency}
\end{equation}
Thus, for a given $PPV_t$, the corresponding group-specific $FOR_{A=a}$ just depends on $n_{A=a|D=1}$, because $n_{A=a}$ and $BR_{A=a}$ are given by the group $a$'s population.
For groups with different probability distributions, $FOR_{A=0}$ and $FOR_{A=1}$ are usually different if just PPV parity is enforced.
Hence, to satisfy sufficiency, at least one of the two groups must deviate from their optimal solution (under PPV parity) to ensure that the FORs of the two groups are equal.
Most importantly, this deviation must not change the group's PPVs so that the PPV parity constraint still holds (with $PPV=PPV_t$).
Let the \textit{solution space} consist of all combinations of $PPV$ and $FOR$ that can be achieved by all groups, based on the groups' probability distributions.
We now show how this solution space can be constructed for one or for more groups.

As shown in Equation~\ref{equation:PPV-parity-BR-equation}, the $PPV$ and the $FOR$ always lie on different sides of the BR, because $n_{A=a} = n_{A=a|D=0} + n_{A=a|D=1}$ and $BR_{A=a}, FOR_{A=a}, PPV_t \in [0,1]$.
Therefore, if $PPV > BR_{A=a}$, the group's $FOR_{A=a}$ must take a value \textit{below} $BR_{A=a}$ and vice versa.
Let $F_a(PPV_{A=a})$ be a group-specific function defined as a group $a$'s $FOR_{A=a}$ that results from maximizing $n_{A=a|D=1}$ for a specific value of $PPV$.
As shown in the proof of Theorem~\ref{theorem:mainresult_optimalsolutionwiththresholdrule}, varying the number of selected individuals without changing the group's PPV lets us specify the range of values a group's FOR can take.
In this way, we can derive the range of values the $FOR$ can take for any $PPV$, which will then allow us to construct the solution space.

In Figure~\ref{fig:theoretical_sufficiency_1group}, the solution space is represented as a white area and the function $F_a(PPV_{A=a})$ is illustrated with a blue line.
For example, for a given $PPV^{\prime}$, point A is achieved by selecting just one individuals with a probability $p_i = PPV$, point B is achieved by maximizing $n_{A=a|D=1}$.
The green line in Figure~\ref{fig:theoretical_sufficiency_1group} visualizes the combinations resulting from applying optimal decision rules for a specific $PPV$:
As we stated in Theorem~\ref{theorem:mainresult_optimalsolutionwiththresholdrule}, it is optimal to apply a lower-bound threshold and if $PPV \in [BR_{A=a},1]$ and an upper-bound threshold is optimal if $PPV \in [PPV_0, BR_{A=a}]$, where $PPV_0$ denotes the $PPV$ for which $\alpha PPV + \beta (1-PPV) = 0$.
If $PPV<PPV_0$, it is optimal to minimize the number of selected individuals (see Section~\ref{ssec:ppv_parity}).
The intuition to construct a solution that satisfies sufficiency is the following: under PPV parity, for a given $PPV^{\prime}$, it is optimal to apply a decision rule leading to a PPV-FOR combination lying at point B.
However, the FOR that this decision rule yields might not lie within the other group's solution space, making a deviation in point A necessary.
\begin{figure*}[t]
\centering
\begin{subfigure}[t]{0.507\textwidth}
    \centering
    \includegraphics[width=\textwidth]{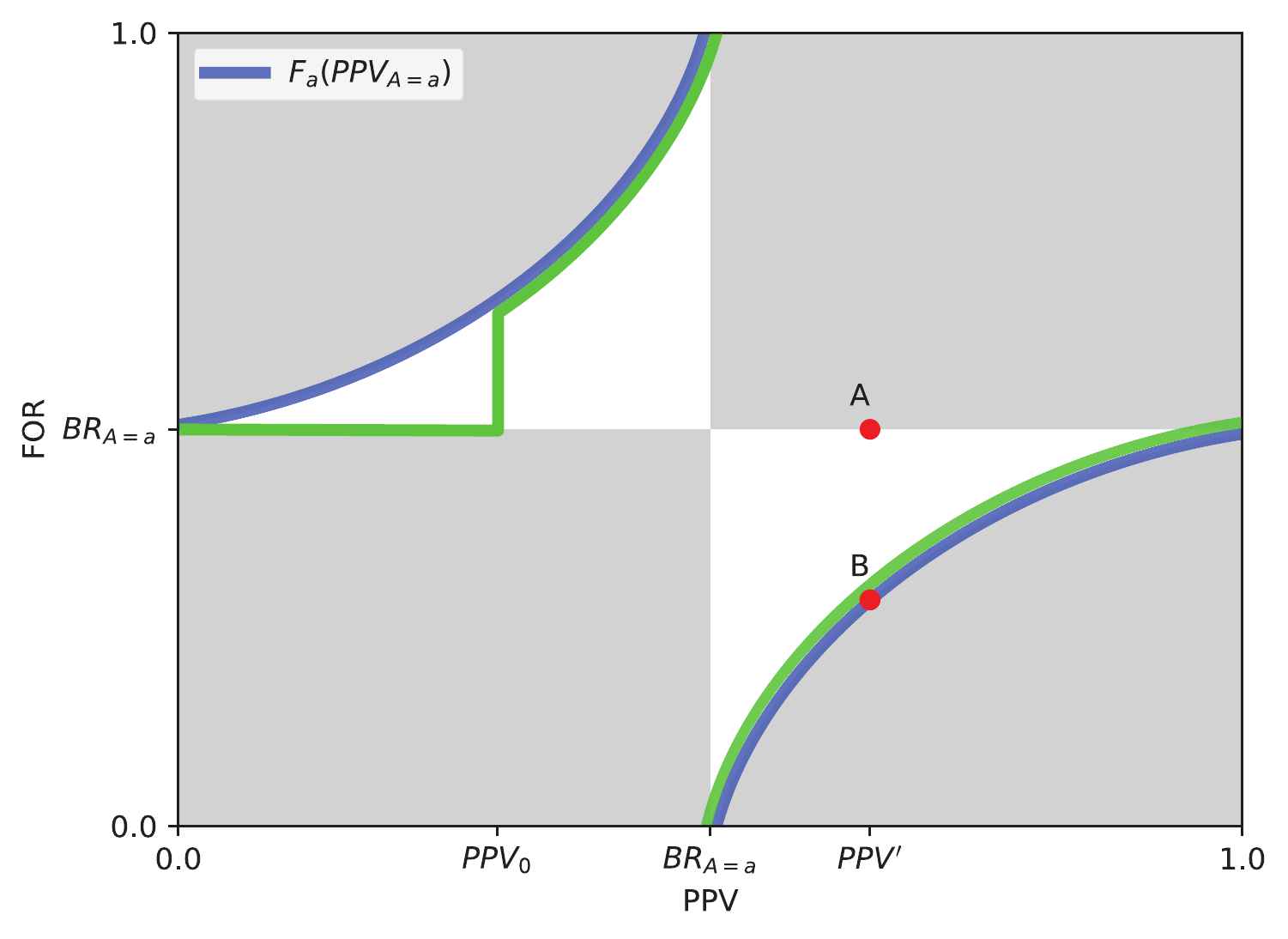}
    \caption{Solutions space: possible (white area) and optimal (green line) \\PPV-FOR combinations for one group}
    \label{fig:theoretical_sufficiency_1group}
\end{subfigure}
\hfill
\begin{subfigure}[t]{0.487\textwidth}
    \centering
    \includegraphics[width=\textwidth]{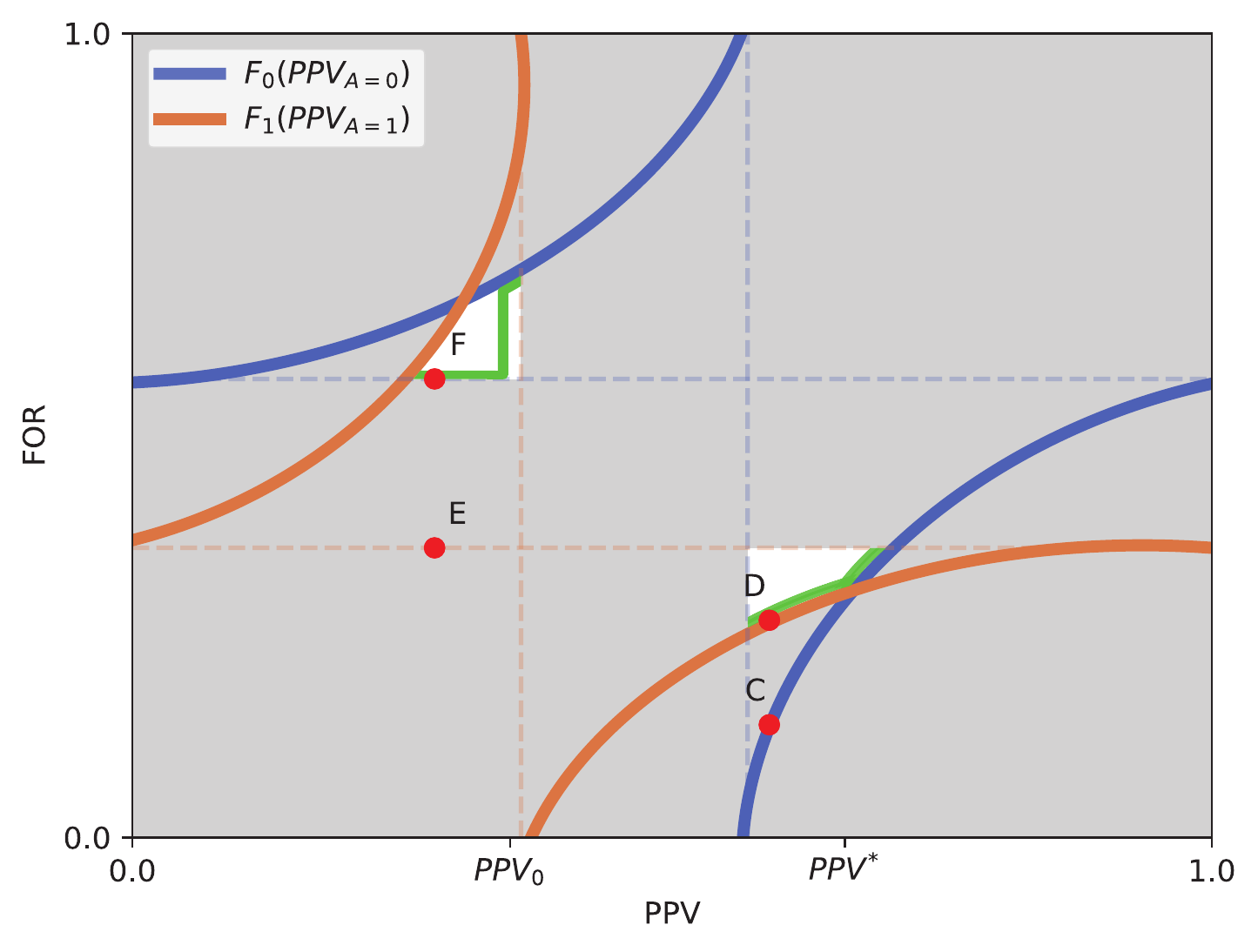}
    \caption{Overlying solutions spaces: possible (white area) and optimal \\(green line) PPV-FOR combinations for two groups}
    \label{fig:theoretical_sufficiency_2groups}
\end{subfigure}
\caption{PPV-FOR plot: utility-maximizing PPV-FOR combinations for specific values of $PPV$}
\label{fig:theoretical_sufficiency}
\end{figure*}

Let us now generalize this to two (or more) groups.
To construct the joint solution space of several groups, the individual solution spaces can be laid on top of each other.
Figure~\ref{fig:theoretical_sufficiency_2groups} illustrates this for two groups, 0 (blue) and 1 (orange).
The two white areas include all PPV-FOR combinations that are feasible for both groups.
Inside this resulting smaller solution space, the optimal $FOR$ for each possible $PPV$ can be found (as visualized with the green line in Figure~\ref{fig:theoretical_sufficiency_2groups}), which satisfies sufficiency.
Enforcing PPV parity does not result in a solution that also satisfies FOR parity simply by chance, apart from one exceptional case: That is, only if $PPV = PPV^{*}$, where $PPV^{*}$ denotes the specific $PPV$ for which the two groups' lines representing their optimal decision satisfying PPV parity (i.e., $F_0(PPV_{A=0})$ and $F_1(PPV_{A=1})$) intersect, the decision rule satisfying PPV parity also satisfies sufficiency.
If $PPV \neq PPV^{*}$, one of the two groups must deviate from their optimal PPV-FOR combination in order to match the other group's $FOR$ and to ensure that not only PPV parity but also FOR parity is satisfied.
Visually, this deviation (representing a change in the FOR for a remaining value of PPV) can be perceived as a vertical move away from the optimal PPV-FOR combination (satisfying PPV parity) towards the edge of the solution space (see $C \rightarrow D$ or $E \rightarrow F$ in Figure~\ref{fig:theoretical_sufficiency_2groups}).

The construction of the solution space (as visualized in Figure~\ref{fig:theoretical_sufficiency_2groups}) directly generalizes to cases with any number of groups, i.e., cases in which the sensitive attribute is a set consisting of more than two different values.
Theorem~\ref{theorem:mainresult_sufficiency} shows that this makes a full satisfaction of within-group fairness impossible.
\begin{theorem}
\label{theorem:mainresult_sufficiency}
Optimal decision rules $d^{\ast}$ that satisfy sufficiency lead to within-group unfairness in all but one of the groups if a solution exists.
\end{theorem}
\begin{proof}
Let us first consider a binary protected attribute $A$.
The intersection of the group-specific solution spaces defines all PPV-FOR combinations for which a solution exists.
If the deviating group's $FOR_{A=a}>BR_{A=a}>PPV$, their $FOR_{A=a}$ must match the other group's BR ($E \rightarrow F$ in Figure~\ref{fig:theoretical_sufficiency_2groups}).
Otherwise, if $FOR_{A=a}<BR_{A=a}<PPV$, their $FOR$ must match the other group's $F_a(PPV_{A=a})$ ($C \rightarrow D$ in Figure~\ref{fig:theoretical_sufficiency_2groups}).
This deviation is necessary to satisfy sufficiency and can be achieved by adjusting $n_{A=a|D=1}$.
This represents an equivalent problem as maximizing the utility under PPV parity with case I -- as we discussed it in Section~\ref{ssec:ppv_parity}.
Hence, the deviating group's optimal decision rule can take many forms -- e.g., one could apply a stochastic decision rule that flips a coin to set $D=1$ with probability $q$ for all individuals with $p>\tau_a$, where $\tau_a$ is a group-specific constant.
However, instead of a simple lower- or upper-bound threshold but, are more complex decision rule is required in order to ensure that the correct number of individuals are selected.
Thus, this always leads to unfairness \textit{within} this group to achieve sufficiency \textit{between} the groups: $\exists i,j \in S_a (p_i > p_j \land d_i=0 \land d_j=1)$.
\end{proof}
Notice that any PPV-FOR combination lying inside the solution space but not at the edge is Pareto dominated because there is another point with the same $PPV$ that results in a higher utility.%
\footnote{If $FOR_{A=a}<BR_{A=a}<PPV$, this point lies on one of the groups' $F_a(PPV_{A=a})$, else, this point is situated on one of the groups' BR.
}
The green line in Figure~\ref{fig:theoretical_sufficiency_2groups} represents the optimal PPV-FOR combinations for specific values of $PPV$.
Any number of solution spaces can be laid on top of each other, which is why this finding extends directly to non-binary sensitive attributes.
Though, the more groups are considered (assuming that the groups' $F_a$ functions and their BRs differ), the smaller the solution space becomes.
And, the smaller the solution space, the bigger the required deviation, which produces more within-group unfairness.
An area of size 0 is possible and would imply that sufficiency cannot be satisfied.

Finally, as we can compute the utility resulting from applying an optimal decision rule satisfying sufficiency for any value of $PPV$, we can solve the constrained maximization problem stated in Equation~\ref{equation:optimizationproblem_sufficiency} by choosing the optimal PPV-FOR combination (i.e., the optimal point lying on the green line in Figure~\ref{fig:theoretical_sufficiency_2groups}).

\section{Illustrative Examples}
\label{sec:IllustrativeExamples}

We now illustrate the solutions (that we derived theoretically in the previous section) to showcase the decisions that result from a utility-maximizing decision maker who wants to satisfy different fairness constraints (PPV parity, FOR parity, sufficiency).
First, we demonstrate the form that these optimal decision rules take for different synthetic populations.
Second, we apply the solutions to real data.%
\footnote{Data and code to reproduce our results are available at \href{https://github.com/joebaumann/fair-prediction-based-decision-making}{https://github.com/joebaumann/fair-prediction-based-decision-making}.}

To present our results, we use a simple tuple notation $(\tau_1,\tau_2)$ (where $\tau_1$ denotes the lower- and $\tau_2$ the upper-bound), meaning that any individuals with a probability $p \in [\tau_1,\tau_2]$ is assigned the decision $D=1$ and $D=0$ otherwise.

\subsection{Synthetic Data Example}
\label{subsec:Synthetic-Data-Example}

For three different populations, all of which are composed of two groups (1 and 2) of individuals with probabilities drawn from a Beta distribution, we investigate the form the optimal fair decision rules take.
Table~\ref{tab:synthetic-data-and-solutions} list the detailed parameters for all populations.
Notice that the groups are equal in size in populations 1 and 2, but in population 3, group 1 is much smaller (just 10\% the size of group 0).
In all populations, group 0 is disadvantaged, meaning that it has a lower base rate (BR) than group 1: $BR_{A=0} < BR_{A=1}$ (just slightly lower in population 1, substantially lower in populations 2 and 3).
We present the solutions for decision rules that satisfy a fairness constraint (PPV parity, FOR parity, or sufficiency) while optimizing the decision maker's utility%
\footnote{
This hypothetical utility function represents a situation where a successful loan makes 7, but a default costs the bank 3.
}, which is defined as follows for all three populations:
\begin{equation}
    U=\sum_{i \in S} u_i, \;\; \text{ for } \;\; u_i=\begin{cases}
        7 p_i - 3 (1-p_i), & \text{for $D=1$}\\
        0, & \text{for $D=0$}
    \end{cases}
\label{equation:synthetic-example-utility-function}
\end{equation}
Hence, an individual's expected utility depends on the estimated repayment probability $p$.
Absent any fairness constraint, it is optimal for the bank to grant a loan to all individuals whose $p>t_0=0.3$ (as indicated with the red dashed line in the Figures~\ref{fig:PPV_parity-Population-1}-\ref{fig:PPV_parity-Population-3}).
\begin{table*}[t]
\centering
\caption{Parameters and solutions of the synthetic data example. The acronyms stand for group size ($n$), group distribution ($P$), base rate $BR$ (which results from $n$ and $P$), optimal threshold ($t_0$) and resulting PPV ($PPV_{t0}$) for unconstrained utility maximization, optimal thresholds ($t_1,t_2$) and resulting PPV ($PPV_{t0,t1}$) for utility maximization under PPV parity.}
\label{tab:synthetic-data-and-solutions}
\begin{tabular}{lll|cccccc|}
\cline{4-9}
                                                 &                                                  &                 & \multicolumn{2}{c|}{\textbf{Population 1}}                             & \multicolumn{2}{c|}{\textbf{Population 2}}                        & \multicolumn{2}{c|}{\textbf{Population 3}}   \\ \cline{4-9} 
\multicolumn{1}{c}{}                             &                                                  &                 & \multicolumn{1}{c|}{Group 0}         & \multicolumn{1}{c|}{Group 1}    & \multicolumn{1}{c|}{Group 0}    & \multicolumn{1}{c|}{Group 1}    & \multicolumn{1}{c|}{Group 0}    & Group 1    \\ \hline
\multicolumn{2}{|l|}{\multirow{3}{*}{parameters}}                                                   & $n$             & \multicolumn{5}{c|}{20,000}                                                                                                                                                  & 2,000      \\ \cline{3-9} 
\multicolumn{2}{|l|}{}                                                                              & $P$             & \multicolumn{1}{c|}{Beta(1.9, 1.35)} & \multicolumn{1}{c|}{Beta(3, 2)} & \multicolumn{1}{c|}{Beta(2, 3)} & \multicolumn{1}{c|}{Beta(3, 2)} & \multicolumn{1}{c|}{Beta(2, 3)} & Beta(3, 2) \\ \cline{3-9} 
\multicolumn{2}{|l|}{}                                                                              & $BR$            & \multicolumn{1}{c|}{0.58}            & \multicolumn{1}{c|}{0.60}       & \multicolumn{1}{c|}{0.39}       & \multicolumn{1}{c|}{0.60}       & \multicolumn{1}{c|}{0.39}       & 0.60       \\ \hline
\multicolumn{1}{|l|}{\multirow{4}{*}{solutions}} & \multicolumn{1}{l|}{unconstr.}                   & $t_0$           & \multicolumn{6}{c|}{0.30}                                                                                                                                                                  \\ \cline{3-9} 
\multicolumn{1}{|l|}{}                           & \multicolumn{1}{l|}{}                            & $PPV_{t_0}$     & \multicolumn{1}{c|}{0.65}            & \multicolumn{1}{c|}{0.63}       & \multicolumn{1}{c|}{0.51}       & \multicolumn{1}{c|}{0.63}       & \multicolumn{1}{c|}{0.51}       & 0.63       \\ \cline{2-9} 
\multicolumn{1}{|l|}{}                           & \multicolumn{1}{l|}{\multirow{2}{*}{PPV parity}} & $(t_1, t_2)$    & \multicolumn{1}{c|}{(0.27, 1)}       & \multicolumn{1}{c|}{(0.33, 1)}  & \multicolumn{1}{c|}{(0.44, 1)}  & \multicolumn{1}{c|}{(0.08, 1)}  & \multicolumn{1}{c|}{(0.37, 1)}  & (0, 0.84)  \\ \cline{3-9} 
\multicolumn{1}{|l|}{}                           & \multicolumn{1}{l|}{}                            & $PPV_{t_1,t_2}$ & \multicolumn{2}{c|}{0.64}                                              & \multicolumn{2}{c|}{0.60}                                         & \multicolumn{2}{c|}{0.56}                    \\ \hline
\end{tabular}
\end{table*}
\begin{figure}
\centering
\begin{subfigure}{0.37\textwidth}
    \vspace*{3mm}
    \centering
    \includegraphics[width=\textwidth]{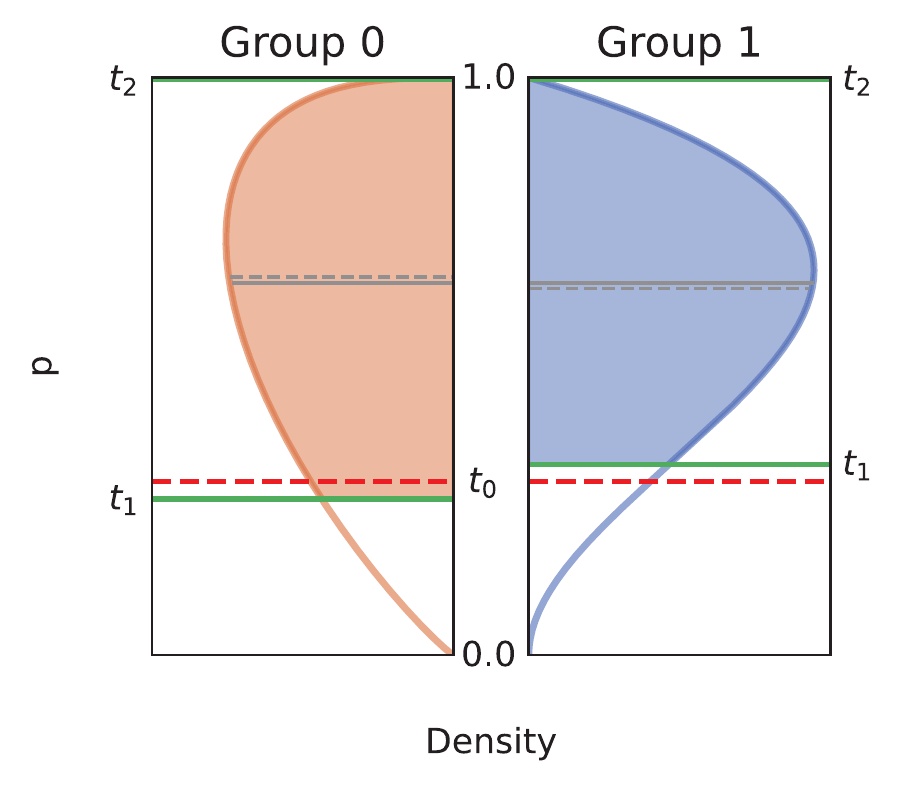}
    \caption{Population 1}
    \label{fig:PPV_parity-Population-1}
    \vspace*{5mm}
\end{subfigure}
\hfill
\begin{subfigure}{0.37\textwidth}
    \centering
    \includegraphics[width=\textwidth]{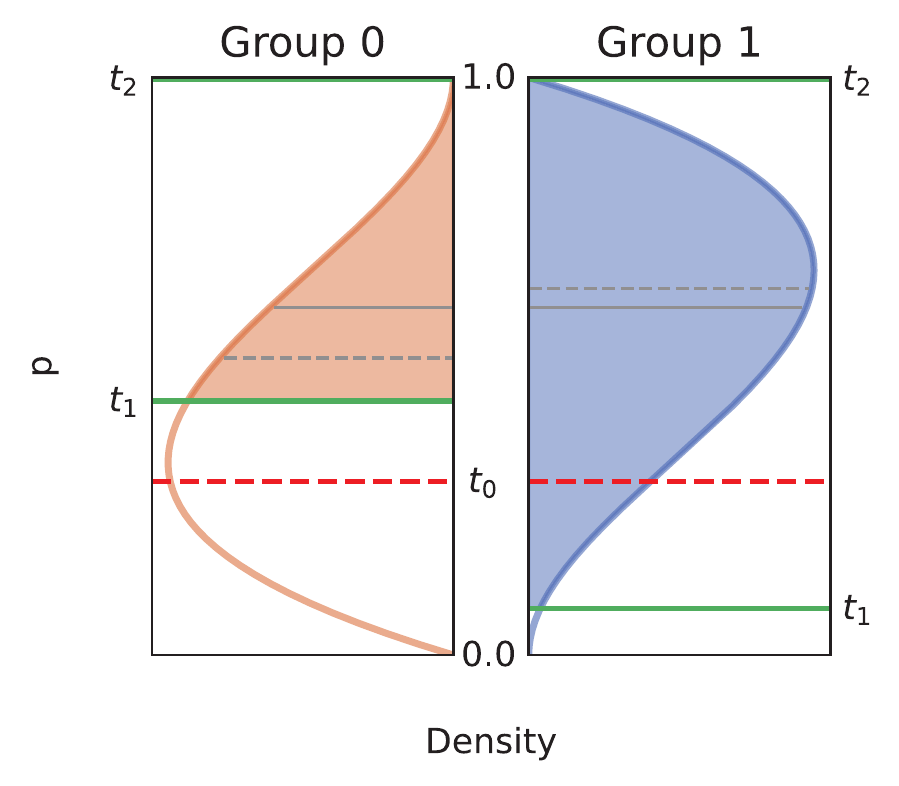}
    \caption{Population 2}
    \label{fig:PPV_parity-Population-2}
    \vspace*{5mm}
\end{subfigure}
\hfill
\begin{subfigure}{0.37\textwidth}
    \centering
    \includegraphics[width=\textwidth]{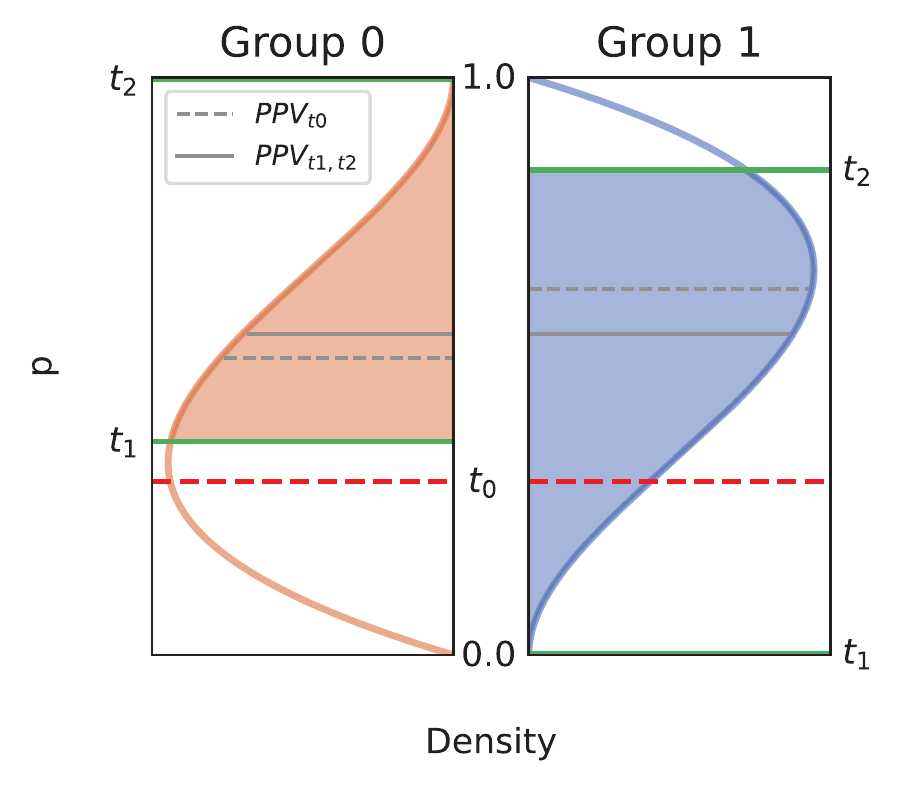}
    \caption{Population 3}
    \label{fig:PPV_parity-Population-3}
\end{subfigure}
\caption{Utility maximization under PPV parity (synthetic)}
\label{fig:PPV_parity}
\end{figure}

Figure~\ref{fig:PPV_parity} visualizes the probability densities of $p$ along with the optimal decision rules under PPV parity, which are different for each of the three populations.
Applying a single threshold $t_0$ results in unequal $PPVs$ for all of the three cases (see $PPV_{t0}$ in Figure~\ref{fig:PPV_parity}).
The solid green lines indicate the thresholds $t_1$ and $t_2$ that correspond to the optimal decision rule while satisfying PPV parity.
With this fairness constraint, $PPVs$ are equalized (see $PPV_{t1,t2}$ in Figure~\ref{fig:PPV_parity}).
But, the optimal decision rule used to achieve this depends on the population:
\begin{itemize}
    \item Population 1: Compared to the optimal solution without fairness ($t_0$), group 0's threshold is decreased while group 1's threshold is increased ($t_1^{group \: \mathit{0}} < t_0 < t_1^{group \: \mathit{1}}$) in order to equalize the two groups' $PPVs$.%
    \footnote{
    This result is not surprising as it is conceptually equivalent to solutions for other group fairness metrics.}
    \item Population 2: Unlike in population 1, in population 2, group 0's $PPV_{t0}$ is lower than the one of group 1.
    This means that the disadvantaged group 0 is held to a higher standard ($t_1^{group \: \mathit{0}} > t_0 > t_1^{group \: \mathit{1}}$) to satisfy PPV parity while maximizing the utility.
    This result is likely to occur in practice because, with the single threshold ($t_0$) rule that is used without any fairness constraint, the disadvantaged group's $PPV_{t0}$ is lower for groups with similar distributions.
    \item Population 3: Due to the mere difference in the group sizes (all else equal to population 2), it is much more ``costly'' to change group 0's threshold (relatively to group 1).
    Thus, in this situation, it is optimal to deviate less from group 0's unconstrained optimum.
    This results in an optimal $PPV=0.56$, which is lower than $BR_{A=1}$.
    For this reason, it is optimal to apply an upper-bound threshold for group 1 (set $t_1^{group \: \mathit{1}}=0$ and $t_2^{group \: \mathit{1}}<1$), i.e., deliberately disregarding those individuals with the highest probability of belonging to the positive class $Y=1$.
    This leads to an extreme form of within-group unfairness.
    It means that a utility-maximizing decision maker would ``sacrifice'' the best individuals (with a probability between $t_2^{group \: \mathit{1}}$ and $1$) of the smaller group 1 in favor of ``keeping'' individuals with a probability slightly above $t_2^{group \: \mathit{0}}$ in the bigger group 0.
    In the loan granting scenario, this would imply \textit{not} granting a loan to those individuals of group 1 that are most likely to repay.
    At the same time, group 1's individuals with the lowest repayment probability (i.e., those with a high probability of default) are granted a loan.
\end{itemize}
This example shows clearly that the optimal decision rules depend on the groups' probability distributions.
In some cases, this can lead to counter-intuitive solutions: it is possible that the disadvantaged group is held to a higher standard or that the most promising individuals of the advantaged group are omitted.

We present additional results (i.e., optimal decision rules under FOR parity and under sufficiency) for the synthetic data example in Appendix~\ref{appendix-sec:synthetic}.

\subsection{Real-World Example: COMPAS}

We now illustrate our results for the recidivism prediction case, using the ProPublica recidivism dataset%
\footnote{
We used the already pre-processed dataset named ``\textit{propublica-recidivism\_numerical.csv},'' which can be accessed here: \url{https://github.com/algofairness/fairness-comparison/tree/master/fairness/data/preprocessed}.
A detailed description of the COMPAS dataset and the use case is provided by~\cite{friedler2019comparative} and~\cite{angwin2016machine}.
}, which includes data from the COMPAS tool collected by~\cite{angwin2016machine}.
We trained a logistic regression (based on the implementation by scikit-learn~\cite{scikit-learn}) to predict probabilistic recidivism risk scores (achieving an overall accuracy of 0.69).

A decision maker has to transfer a risk score into a decision.
This involves weighing the severity of FP and FN in the utility function.
We present the utility-maximizing solutions for three possible settings, each one specified by different utility weights FP and FN, while TP=TN=1 is kept constant. These different utility functions are paired with different fairness requirements (no fairness constraint, PPV parity, and FOR parity) w.r.t. the protected attribute \textit{race}, which can take two values, \textit{Caucasian} ($c$) or \textit{non-Caucasian} ($nc$).
The class $Y=1$ denotes a recidivist, and each individual must either be detained ($D=1$) or released ($D=0$).
Figure~\ref{fig:COMPAS-score-distributions} shows the score distributions of the two groups.
The base rate of non-Caucasians (0.49) is higher than the one of Caucasians (0.4), indicating that non-Caucasians more likely to be predicted as being of high risk to recidivate, on average. 
The specified utility weights and the resulting optimal decisions for the different fairness requirements are presented in Table~\ref{tab:COMPAS-results-extended}.
The (un)constrained optimal decision rules differ largely across the three cases:
\begin{itemize}
    \item \textit{Case 1} represents a situation where a decision maker is indifferent about what is worse: incorrectly classifying an innocent person as guilty or releasing a defendant who goes on to recidivate.
    Thus, equal weights for FP and FN are chosen.
    For such a case, a lower-bound threshold of $t_{u1}=0.5$ is optimal from the decision maker's perspective.
    However, the two fairness metrics are not just satisfied by chance, because this threshold leads to different FORs and PPVs for the the two groups ($PPV_{A=c} < PPV_{A=nc}$ and $FOR_{A=c} < FOR_{A=nc}$).
    \item \textit{Case 2} showcases decision rules representing a shift towards protecting the innocent, therefore, using a much lower weight (-10) for FP.
    For the unconstrained setting, this results in fewer detained individuals overall, with an optimal lower-bound threshold of $t_{u2}=0.85$.
    As the two groups' distributions are similar above this threshold, their PPVs are almost the same, which is why just a slight adjustment of the group-specific thresholds is needed to satisfy PPV parity.
    In contrast, very different group-specific thresholds are optimal to satisfy FOR parity.
    Due to the lower BR of the non-Caucasian group (see the right-skewed distribution in Figure~\ref{fig:COMPAS-score-distributions}), it is optimal to release all Caucasians with a risk score below 0.98.
    This makes sure that released individuals are equally likely to recidivate across groups.
    \item \textit{Case 3} resembles a decision maker who cares more about punishing guilty than protecting innocent individuals, which is represented with a large negative value for FN.
    Absent any fairness constraint, this results in a lower optimal lower-bound threshold ($t_{u3}=0.15$), leading to more overall detentions.
    As opposed to case 2, this results in almost equal FORs (because the two groups' distributions are similar below the unconstrained threshold) but the two groups' PPVs differ largly.
    To satisfy PPV parity, it is optimal to detain almost all non-Caucasians (those with a risk score above $\tau_1 = 0.05$%
    \footnote{
    If the non-Caucasian group were much smaller, this would result in an upper-bound threshold, i.e., the non-Caucasian with the highest recidivism risk would be released -- which is similar to group 1 in the population 3's result in the synthetic example (see Figure~\ref{fig:PPV_parity-Population-3}).
    }) while detaining a much smaller fraction of Caucasians ($\tau_1 = 0.27$).
\end{itemize}
\begin{figure*}
  \centering
  \includegraphics[width=0.45\textwidth]{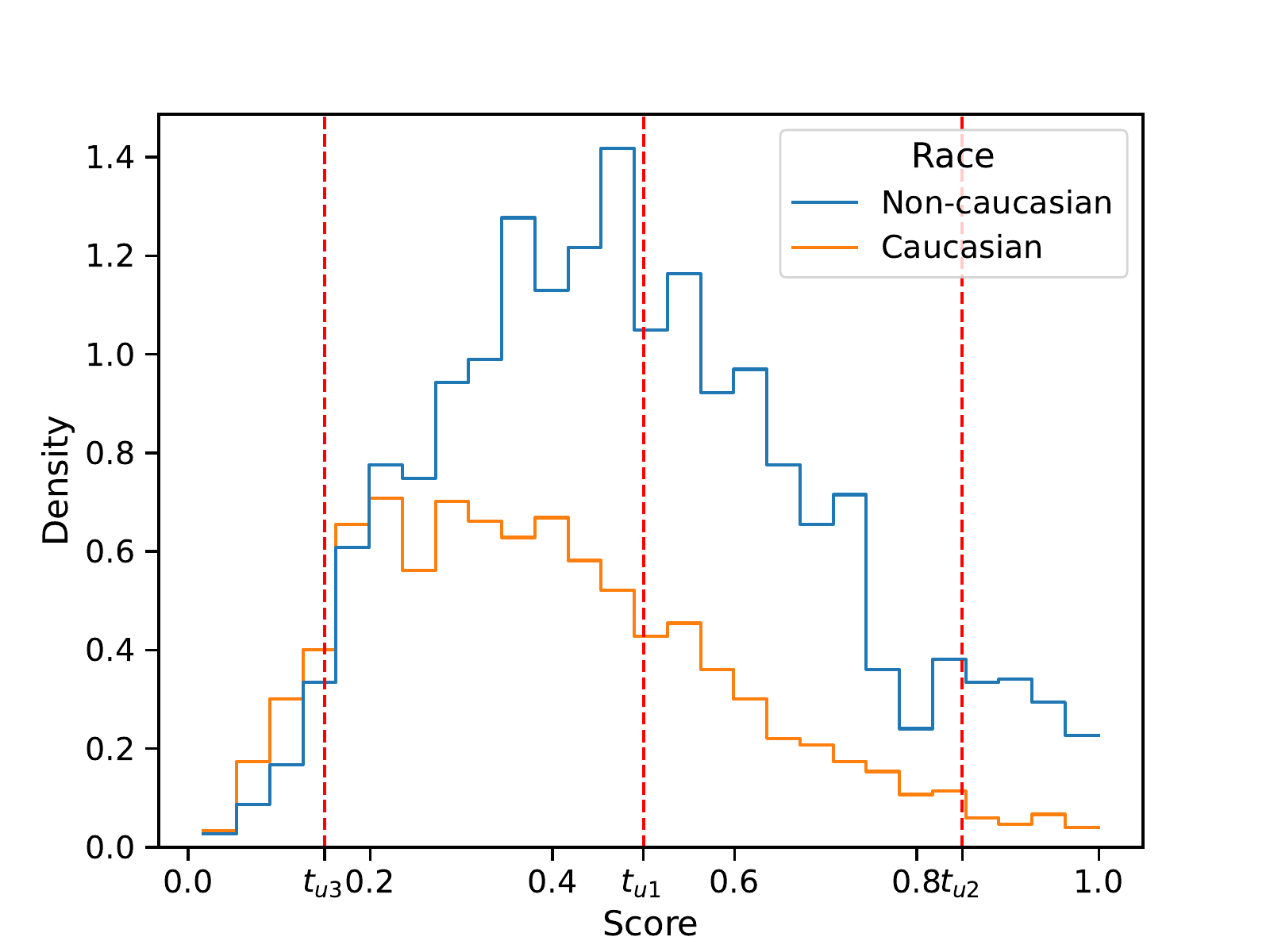}
  \caption{Score distributions by race and optimal unconstrained decision rules ($t_u$) for different utility functions (COMPAS)}
    \label{fig:COMPAS-score-distributions}
\end{figure*}
\begin{table*}
\centering
\caption{Optimal decision rules (COMPAS) for utility functions with different weights (TP, FP, FN, TN) paired with different fairness requirements (no fairness constraint, PPV parity, and FOR parity). The acronyms stand for base rate (BR), the optimal threshold ($t_u$) for unconstrained utility maximization, and the optimal thresholds ($t_1,t_2$) for utility maximization under fairness.}
\label{tab:COMPAS-results-extended}
\begin{tabular}{|l|l|c|c|c|c|c|c|} 
\cline{3-8}
\multicolumn{1}{l}{}        & \begin{tabular}[c]{@{}l@{}}\\\end{tabular} & \multicolumn{2}{c|}{\textbf{Case 1}}      & \multicolumn{2}{c|}{\textbf{Case 2}}      & \multicolumn{2}{c|}{\textbf{Case 3}}       \\ 
\cline{2-8}
\multicolumn{1}{l|}{}       &                                            & \textit{Caucasian} & \textit{non-Caucas.} & \textit{Caucasian} & \textit{non-Caucas.} & \textit{Caucasian} & \textit{non-Caucas.}  \\ 
\cline{2-8}
\multicolumn{1}{l|}{}       & BR                                         & 0.40               & 0.49                 & 0.40               & 0.49                 & 0.40               & 0.49                  \\ 
\cline{2-8}
\multicolumn{1}{l|}{}       & TP, FP, FN, TN                             & \multicolumn{2}{c|}{1, -1, -1, 1}         & \multicolumn{2}{c|}{1, -10, -1, 1}        & \multicolumn{2}{c|}{1, -1, -10, 1}         \\ 
\hline
\multirow{3}{*}{unconstr.}  & $t_u$                                      & \multicolumn{2}{c|}{$t_{u1}$ = 0.50}       & \multicolumn{2}{c|}{$t_{u2}$ = 0.85}      & \multicolumn{2}{c|}{$t_{u3}$ = 0.15}       \\ 
\cline{2-8}
                            & PPV                                        & 0.65               & 0.68                 & 0.92               & 0.92                 & 0.42               & 0.50                  \\ 
\cline{2-8}
                            & FOR                                        & 0.30               & 0.34                 & 0.38               & 0.46                 & 0.11               & 0.11                  \\ 
\hline
\multirow{3}{*}{PPV parity} & ($t_1,t_2$)                                & (0.52, 1)          & (0.49, 1)            & (0.84, 1)          & (0.85, 1)            & (0.27, 1)          & (0.05, 1)             \\ 
\cline{2-8}
                            & PPV                                        & \multicolumn{2}{c|}{0.67}                 & \multicolumn{2}{c|}{0.92}                 & \multicolumn{2}{c|}{0.49}                  \\ 
\cline{2-8}
                            & FOR                                        & 0.30               & 0.33                 & 0.38               & 0.46                 & 0.18               & 0.03                  \\ 
\hline
\multirow{3}{*}{FOR parity} & ($t_1,t_2$)                                & (0.57, 1)          & (0.47, 1)            & (0.98, 1)          & (0.62, 1)            & (0.16, 1)          & (0.15, 1)             \\ 
\cline{2-8}
                            & PPV                                        & 0.71               & 0.66                 & 0.99               & 0.76                 & 0.43               & 0.50                   \\ 
\cline{2-8}
                            & FOR                                        & \multicolumn{2}{c|}{0.32}                 & \multicolumn{2}{c|}{0.40}                 & \multicolumn{2}{c|}{0.11}                  \\
\hline
\end{tabular}
\end{table*}
Without fairness-enforcing restrictions, the same prediction model can turn out to be fair or unfair, w.r.t. a specific fairness metrics, depending on the utility function.
For example, in case 2, PPV parity is met in the unconstrained case, whereas there is a huge difference in PPVs in case 3.
Note, however, that this cannot be generalized: there is no guarantee that PPV parity or FOR parity are met in the unconstrained case for a given utility function, as this depends on the groups' probability distributions.
Thus, assuming that a prediction model is fair if it meets PPV parity or FOR parity is misleading because this only holds for specific utility functions and probability distributions but not in general.
This contradicts the approach suggested by Northpointe, who claim that PPV and FOR are the only relevant measures to determine the treatment disparity of such a tool for different groups~\cite{Dieterich2016}.
Interestingly, for the COMPAS example, introducing fairness constraints (in the form of PPV parity or FOR parity) leads to a lower group-specific threshold for the non-Caucasians, resulting in a higher fraction of detained individuals for the disadvantaged group -- which is similar to the population 3's result in the synthetic example (see Figure~\ref{fig:PPV_parity-Population-2}).%
\footnote{
There is just one exception to this: enforcing FOR parity in case 2 leads to a slightly higher threshold for the disadvantaged group, which is similar to the population 1's result in the synthetic example (see Figure~\ref{fig:PPV_parity-Population-1}).
}
Further, in some cases, it is optimal to release almost all individuals of the advantaged group or to detain almost all individuals of the disadvantaged group.
This is counter-intuitive as one would expect that introducing a fairness constraint should favor the disadvantaged group.

\section{Conclusions}
\label{sec:Conclusions}

In this paper, we analyze common group fairness metrics that have been proposed to mitigate the unfairness of algorithmic decision making systems.
We formulate algorithmic fairness as a constrained optimization problem representing a decision maker who wants to maximize the total utility while also satisfying a fairness constraint.
A similar solution has been provided by~\cite{hardt2016equality,10.1145/3097983.3098095} for the group fairness metrics (conditional) statistical parity, TPR parity, and FPR parity -- all leading to group-specific lower-bound thresholds.
In contrast to these fairness metrics, we find that for the group fairness metrics PPV parity and FOR parity, optimal decision rules take the form of group-specific lower-bound or upper-bound thresholds.
This is counter-intuitive as it means that, in certain situations, it can be optimal for decision makers to select the `worst' individuals of one group and omit the most promising ones.
In the loan granting scenario, for one of the groups, this would mean that individuals who are most likely to default are granted a loan, whereas those who are most likely to pay back their loan are not granted one.
Similarly, to achieve PPV parity in recidivism risk prediction, it can be optimal to release defendants with the highest recidivism risk in one of the groups.
Additionally, our work shows that there is a trade-off between the group fairness criterion sufficiency and within-group fairness.
Namely, to satisfy sufficiency, it is optimal to sacrifice within-group fairness for all but one of the groups.

Experts increasingly call for fairer algorithms.
Considering these byproducts of the group fairness metrics PPV parity, FOR parity, and sufficiency, we emphasize that these potential consequences must be considered when imposing such fairness criteria on utility-maximizing decision makers.
We hope that our findings foster the discussion of fair algorithmic decision making and, in particular, support policymakers who find themselves in the position where they need to choose a specific definition of fairness.

\begin{acks}
We thank our three anonymous reviewers for their helpful feedback.
This work was supported by Innosuisse -- grant number 44692.1 IP-SBM -- and by the National Research Programme
``Digital Transformation'' (NRP 77) of the Swiss National Science
Foundation (SNSF) -- grant number 187473.
\end{acks}

\bibliographystyle{ACM-Reference-Format}
\bibliography{main}

\appendix

\section{Utility-weighted confusion matrix}
\label{appendix:Utility-weighted-confusion-matrix}

Table~\ref{tab:utilityconfusionmatrix} shows a confusion matrix with the parameters ($u_{11}$, $u_{12}$, $u_{21}$, and $u_{22}$) used to weight the possible outcomes.
\begin{table}[ht]
  \caption{Confusion matrix weighted with utilities}
  \label{tab:utilityconfusionmatrix}
  \begin{tabular}{ccl}
                               & $Y=1$                         & $Y=0$                          \\ \cline{2-3}
    \multicolumn{1}{l|}{$D=1$} & \multicolumn{1}{l|}{$u_{11}$} & \multicolumn{1}{l|}{$u_{12}$}  \\ \cline{2-3}
    \multicolumn{1}{l|}{$D=0$} & \multicolumn{1}{l|}{$u_{21}$} & \multicolumn{1}{l|}{$u_{22}$}  \\ \cline{2-3}
\end{tabular}
\end{table}

\section{Proof of Lemma~\ref{lemma:case_I_PPV}}
\label{appendix:proof-lemma-case_I_PPV}

Recall that the number of individuals that are assigned a positive decision ($n_{D=1}$) can be written as $\sum_{i \in S} d_i$ and that the positive predictive value (PPV) is defined as:
\begin{displaymath}
PPV = P[Y=1|D=1] = \frac{1}{n_{D=1}} \sum_{i \in S} p_i d_i.
\end{displaymath}
Suppose that $n_{D=1}$ is predefined and that the $PPV$ is given.
For binary group membership $A$, the total utility $\tilde U$ can be written as:
\begin{displaymath}
\begin{split}
\tilde U & = \sum_{i \in S} \beta d_i + (\alpha - \beta) \sum_{i \in S} p_i d_i\\
& = \beta n_{D=1} + (\alpha - \beta) (PPV n_{D=1})\\
& = \beta n_{D=1} + \alpha PPV n_{D=1} - \beta PPV n_{D=1} \\
& = (\alpha PPV + \beta (1-PPV) ) n_{D=1}.
\end{split}
\end{displaymath}
Thus, for a given $PPV$ and a predefined selection capacity $n_{D=1}$, the total utility is given, and any decision rule that satisfies the constraint in Equation~\ref{equation:optimizationproblem_equivalent_with_PPV} is optimal.

\section{Optimal Decision Rules under FOR Parity}
\label{appendix:for_parity}

\renewcommand{\theequation}{C.\arabic{equation}}

In this section, we present the optimal solution for a utility-maximizing decision maker that wants to satisfy the group fairness metric false omission rate (FOR) parity.
The FOR is defined as the average probability of individuals with $D=0$ to have $Y=1$, which can be written as $\frac{1}{n_{D=0}} \sum\limits_{i \in S} p_i (1-d_i)$, where $d_i$ is a binary multiplier representing the decision that is made for an individual $i$.
The fairness definition FOR parity requires this value to be the same across groups.
In the following, we interpret the decision problem as a \emph{selection problem}, however, this time denoting individuals with $D=0$ as ``being selected.''

Again, the solution is also composed of two consecutive steps.
First, we derive the optimal decision rules $d^{\ast}$ for a simplified constraint: We assume that the FOR of both groups must be equal to a predefined value $FOR_t \in [0, 1]$.
Then, we solve the full optimization problem by maximizing the decision maker's utility over all possible values of $FOR_t$.

We can translate this optimization problem into an equivalent problem, defining $\hat u_i$ as the {\em relative utility gain} when switching the decision from $D=1$ to $D=0$.
So, $\hat u_i =0$ for $D=1$, and $\hat u_i = -\alpha p_i - \beta (1-p_i)$ for $D=0$.
Thus, the constrained optimization problem has the form:
\begin{equation}
\begin{split}
&\argmax_d \;\;\;\; \hat U=\sum_{i \in S} \hat u_i (1-d_i) \\
&\text{subject to} \;\;\;\; \frac{1}{n_{A=a|D=0}} \sum_{j \in S_a} p_j (1-d_j) = FOR_t, \;\; \text{for } FOR_t \in [0,1] ,
\end{split}
\label{equation:optimizationproblem_equivalent_FOR}
\end{equation}%
where $S_a$ is the set of all individuals of group $a$, $n_{A=a|D=0}$ denotes the number of individuals in group $a$ with $D=0$.
Hence, the constraint describes a parity of the two groups' FORs.
Since the FOR can only be defined if at least one individual is selected, we assume $n_{A=a|D=0}\geq1$ for each group.
The solution for the optimal decision rules while satisfying FOR parity across groups is analogous to the one under positive predictive value (PPV) parity (see Lemma~\ref{lemma:case_I_PPV} and Theorem~\ref{theorem:mainresult_optimalsolutionwiththresholdrule} in the Section~\ref{ssec:ppv_parity}).

We first analyze case I with a simplified fairness constraint, where we assume that the FOR of both groups must be equal to a predefined value between 0 and 1, denoted by $FOR_t$.
Suppose that $n_{D=1}$ is predefined, which is equivalent to $n_{D=0}$ being predefined.
Similar to Lemma~\ref{lemma:case_I_PPV} (along with its proof in Appendix~\ref{appendix:Utility-weighted-confusion-matrix}), the total utility for an optimal solution satisfying FOR parity is:
\begin{equation}
\hat U = (-\alpha FOR_t - \beta (1-FOR_t) ) n_{D=0}.
\label{equation:rearrangedpredefined_FOR}
\end{equation}
Following similar reasoning as in the case of PPV parity, we end up with a conceptually identical solution for case I, where $n_{D=1}$ is predefined.
Namely, for a given $FOR_t$, the total utility $\hat U$ is the same for any solution under FOR parity.
Any decision rule $d(p,a)$ with $n_{D=1}$ that satisfies the constraint stated in Equation~\ref{equation:optimizationproblem_equivalent_FOR} for a given $FOR_t$ is optimal.
We thus end up with two independent selection problems, one for each group, which consists of finding a selection of individuals characterized by the fact that their average probability equals $FOR_t$.
For each group $a$, selections with different numbers $n_{A=a|D=0}$ are possible.
As long as the predefined $n_{D=0}$ is met, the group membership of the selected individuals does not matter for the resulting total utility.
Hence, there may be several solutions to the optimization problem that differ regarding the number of individuals selected per group (i.e., representing different combinations of ($n_{A=0|D=0},n_{A=1|D=0}$)), with $n_{A=0|D=0}+n_{A=1|D=0} = n_{D=0}$.
Note that most of these solutions violate the group fairness metric statistical parity while still meeting the fairness criterion of FOR parity.

\begin{figure*}
\centering
\begin{subfigure}[t]{0.29\textwidth}
    \centering
    \includegraphics[width=\textwidth]{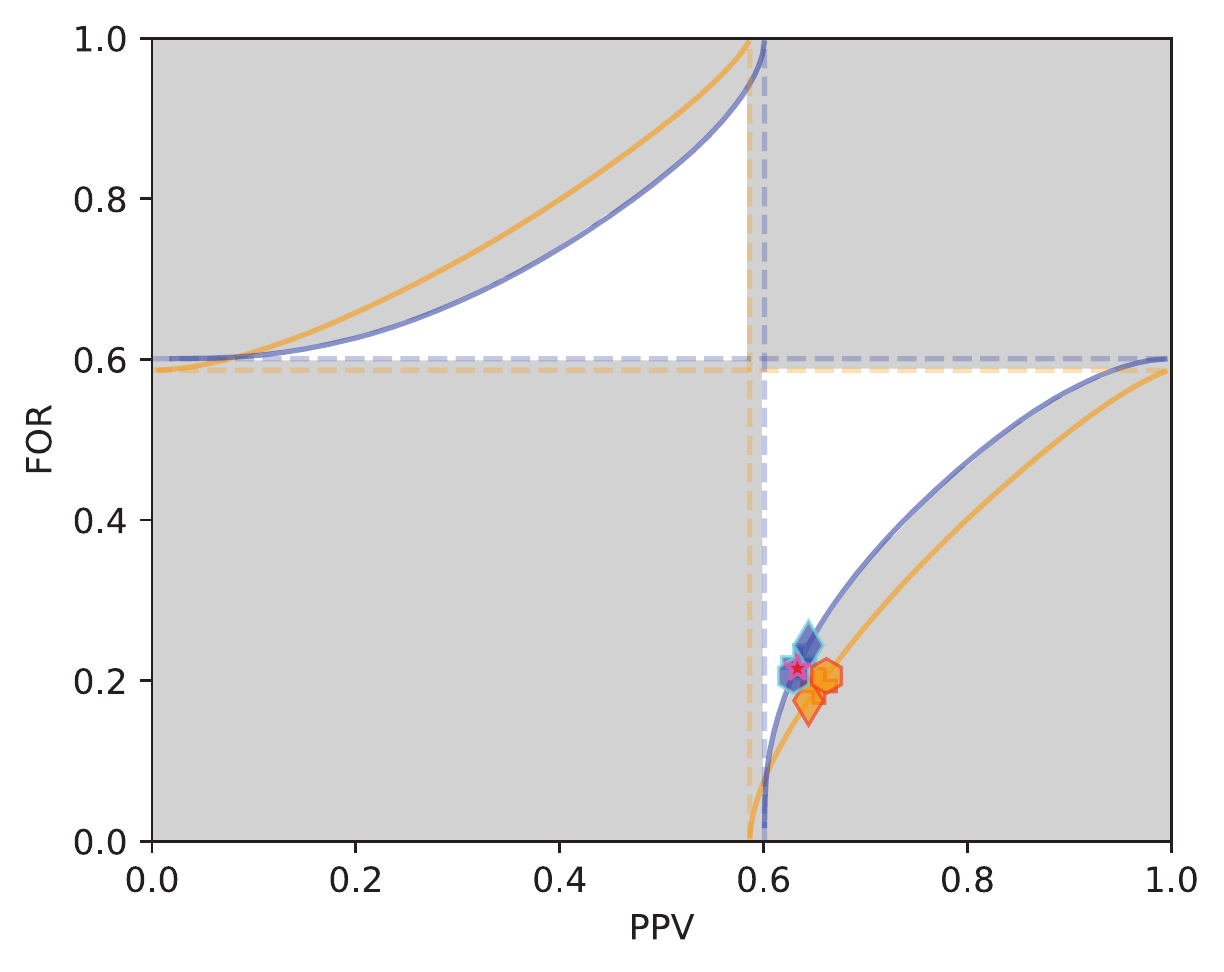}
    \caption{Population 1}
    \label{fig:sufficiency-Population-1}
\end{subfigure}
\hfill
\begin{subfigure}[t]{0.29\textwidth}
    \centering
    \includegraphics[width=\textwidth]{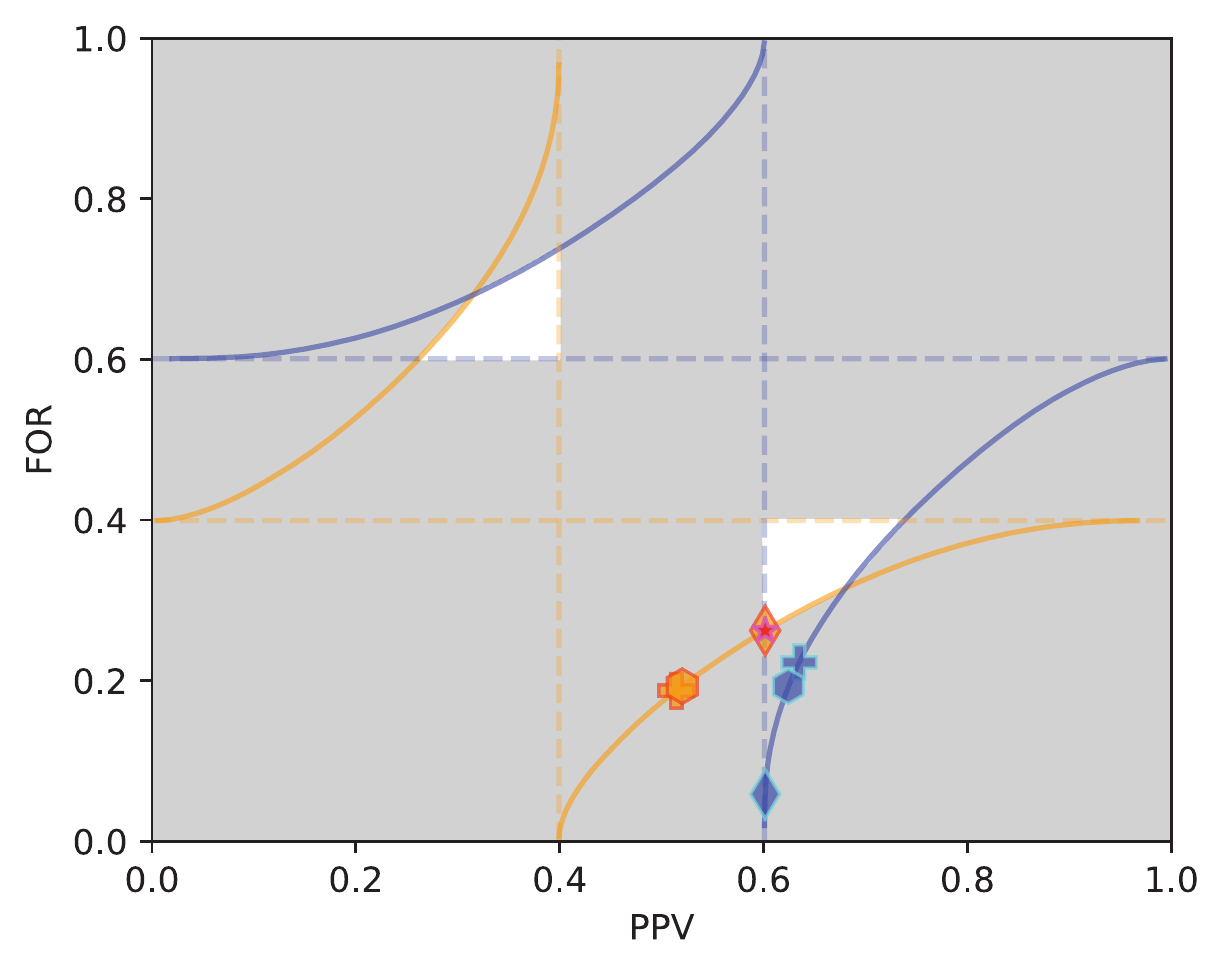}
    \caption{Population 2}
    \label{fig:sufficiency-Population-2}
\end{subfigure}
\hfill
\begin{subfigure}[t]{0.392\textwidth}
    \centering
    \includegraphics[width=\textwidth]{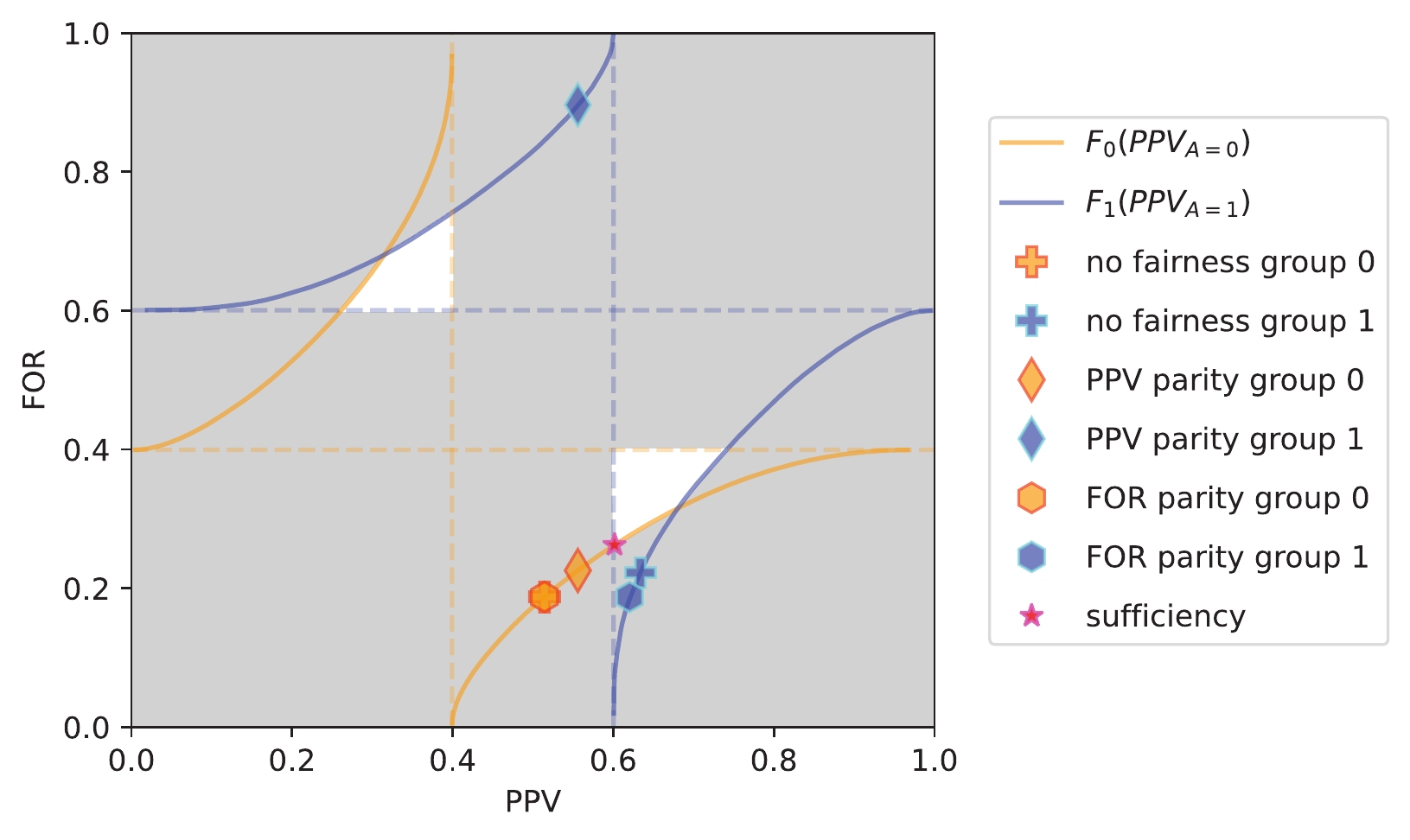}
    \caption{Population 3$\;\;\;\;\;\;\;\;\;\;\;\;\;\;\;\;\;\;\;\;\;\;\;\;$ }
    \label{fig:sufficiency-Population-3}
\end{subfigure}
\caption{Utility maximization under sufficiency or one of its relaxations for three different populations (synthetic)}
\label{fig:sufficiency}
\end{figure*}

\balance

We now analyze case II, where $n_{D=1}$ is not predefined.
Equation~\ref{equation:rearrangedpredefined_FOR} directly shows that for values $FOR_t$ for which $-\alpha FOR_t - \beta (1-FOR_t) <0$, a decision maker who wants to maximize the total utility should minimize $n_{D=0}$, thus assigning the decision $D=0$ only to one individual from each group, yielding a total utility of $\hat U = 2 (-\alpha FOR_t - \beta (1-FOR_t) )$ for a binary protected attribute.
In the following, we thus assume that $-\alpha FOR_t - \beta (1-FOR_t)>0$.
Again we assume that the size of both groups is large but finite.
Above we showed that, under these assumptions, the decision maker's goal is to find the selection that satisfies the constraint $FOR=FOR_t$ with the maximum $n_{D=0}$.
Similarly to the solution for PPV parity (see Theorem~\ref{theorem:mainresult_optimalsolutionwiththresholdrule}), the total utility $\hat U$ is maximized with decision rules $d^{\ast}$ of the following form when applying FOR parity as a fairness constraint:
\begin{equation}
d^{\ast}_i=\begin{cases}
\begin{rcases}
0, & \text{for $p_i \geq \tau_a$} \\
1, & \text{otherwise}\\
\end{rcases}
\text{for $FOR_t>BR_{A=a}$} \\
\begin{rcases}
0, & \text{for $p_i \leq \tau_a$} \\
1, & \text{otherwise}
\end{rcases}
\text{for $FOR_t<BR_{A=a}$},
\end{cases}
\label{equation:d_rule_notpredefined_FOR}
\end{equation}
where $\tau_a$ denote different group-specific constants and $BR_{A=a}$ denotes group $a$'s base rate (BR) which is defined as the ratio of individuals belonging to the positive class ($Y=1$): $BR_{A=a} = P[Y=1|A=a] = \frac{1}{n_{A=a}} \sum_{i \in S_a} p_i$.
The only difference to the solution under PPV parity (see Equation~\ref{equation:d_rule_notpredefined_PPV}) is the fact that decision rules of this form focus on the optimal selection of individuals with $D=0$.

Finally, we perform the second step of the solution: from a discretization of all $FOR$, for which a solution exists, we choose the one that (in combination with the corresponding $n_{D=0}$) maximizes the total utility.
Thereby, every $n_{D=0}$ is composed of the optimal selections $n_{A=a|D=0}$ for all groups $a \in A$, as elaborated in the first step of the solution.

\section{Addendum to the Synthetic Data Example}
\label{appendix-sec:synthetic}

Here we present additional results for the synthetic data example.
Based on the three populations we introduced in Section~\ref{subsec:Synthetic-Data-Example}, we show the resulting PPV and FOR for a utility-maximizing decision maker who want to satisfy PPV parity, FOR parity, or sufficiency.%
\footnote{Data and code to reproduce our results are available at \href{https://github.com/joebaumann/fair-prediction-based-decision-making}{https://github.com/joebaumann/fair-prediction-based-decision-making}.}

Figure~\ref{fig:sufficiency} visualizes the three populations' solution spaces containing possible PPV-FOR combinations and the solutions that maximize utility -- with or without the (relaxed) fairness constraints.
Absent any fairness constraint, it is optimal for the bank to grant a loan to all individuals whose $p>t_0=0.3$.
This is indicated with a group-specific cross (orange for group 0 and blue for group 1) in the Figures~\ref{fig:sufficiency-Population-1}-\ref{fig:sufficiency-Population-3}.

In populations 1 and 2, $PPV > max(BR_{A=a})$ and $FOR < min(BR_{A=a})$, thus, individuals with higher probability are preferred to those with lower probability.
However, in population 3, the upper-bound threshold used for group 1 leads to $PPV_{t1,t2}<BR_{A=1}$ and $FOR_{t1,t2}>BR_{A=1}$ under PPV parity.
The optimal combinations of $PPV$ and $FOR$ lie on $F_0(PPV_{A=0})$, representing optimal PPV-FOR combinations as introduced in Section~\ref{sec:sufficiency}, for group 0 (on $F_1(PPV_{A=1})$ for group 1), both for the solutions satisfying either PPV parity or FOR parity and for the solution without any fairness constraint.
As we can clearly see in all three populations, one group needs to deviate from their optimal PPV-FOR combination to satisfy sufficiency, resulting in within-group unfairness.
For all three populations in Figure~\ref{fig:sufficiency}, the advantaged group 1 is the one that deviates in order to maximize utility while satisfying sufficiency.
However, this is not always the case as it depends on the utility function.%
\footnote{
There are situations in which the disadvantaged group must deviate, as can be seen with those parts of the solution space that are bounded by $F_1(PPV)$.
}

\end{document}